\RequirePackage{fix-cm}
\documentclass[natbib,smallextended]{svjour3}       % onecolumn (second format)
\smartqed  % flush right qed marks, e.g. at end of proof
\usepackage{graphicx}
\usepackage{amsmath}
\usepackage{amssymb}
\usepackage{subcaption}
\usepackage{dsfont}
\usepackage{mathrsfs}
\usepackage[caption=false]{subfig}
 \usepackage{mathptmx}      % use Times fonts if available on your TeX system
 \usepackage{aps-bibstyle}  % use this style if you upload to .tex file only a part of Bibtex created bbl.
\usepackage{latexsym}
\newtheorem{assumption}[theorem]{Assumption}

\newcommand{\du}{\partial}
\newcommand{\R}{\mathbb{R}}

\newcommand{\E}{\mathbb{E}}
\newcommand{\El}{\mathcal{E}}
\newcommand{\de}{\Delta}

\newcommand{\e}{\varepsilon}

\newcommand{\n}{\nabla}

\newcommand{\Lop}{\mathcal{L}}
\newcommand{\Ze}{\mathcal{Z}}

\newcommand{\Lpert}{\tilde{\mathcal{L}}}

\newcommand{\1}{\mathds{1}}
\newcommand{\K}{\mathcal{K}}

\newcommand{\Lopgen}{\mathcal{L}}
\newcommand{\Lope}{\Lopgen}
\newcommand{\LopU}{\mathcal{L}}

\newcommand{\U}{U}

\newcommand{\LopZero}{\mathcal{L}_{0}}
\newcommand{\er}{K_{\mathrm{min}}}
\newcommand{\ef}{K_{\mathrm{max}}}
\newcommand{\Um}{\U}

\newcommand{\Ustd}{U_{\rm{std}}}
\newcommand{\ConstNablaU}{G_{\mr{std}}}
\newcommand{\dspace}{D}%{d_{\rm{space}}}
\newcommand{\ds}{\displaystyle}
\newcommand{\mr}{\mathrm}

\renewcommand{\leq}{\leqslant}

\renewcommand{\geq}{\geqslant}

\newcommand{\dps}{\displaystyle}

 \journalname{my journal}
\title{Error Analysis of Modified Langevin Dynamics}			%\thanks{Grants or other notes
\author{Stephane Redon$^{\dagger}$  \and
        Gabriel Stoltz$^{\ddagger}$ \and Zofia~Trstanova$^{\dagger}$
}

\institute{$^{\dagger}$ \at Inria - Univ. Grenoble Alpes, LJK, F-38000 Grenoble, France\\
CNRS, LJK, F-38000 Grenoble, France\\
              \email{zofia.trstanova@inria.fr}           %  \\
           \and
           $^{\ddagger}$ \at Universit\'e Paris-Est, CERMICS (ENPC), INRIA, F-77455 Marne-la-Vall\'ee, France
}

\date{Received: date / Accepted: date}
\begin{document}

\maketitle

\begin{abstract}
We consider Langevin dynamics associated with a modified kinetic energy
vanishing for small momenta. This allows us to freeze slow particles, and
hence avoid the re-computation of inter-particle forces, which leads to
computational gains. On the other hand, the statistical error may
increase since there are a priori more correlations in time. The aim of
this work is first to prove the ergodicity of the modified Langevin
dynamics (which fails to be hypoelliptic), and next to analyze how the
asymptotic variance on ergodic averages depends on the parameters of
the modified kinetic energy. Numerical results illustrate the approach,
both for low-dimensional systems where we resort to a Galerkin
approximation of the generator, and for more realistic systems using Monte
Carlo simulations.
\end{abstract}

\section{Introduction}
\label{section introduction}

A fundamental purpose of molecular simulation is the computation of macroscopic quantities, typically through averages of functions of the variables of the system with respect to a given probability measure~$\mu$, which defines the macroscopic state of the system. In the most common setting, the probability measure~$\mu$ with respect to which averages are computed corresponds to the canonical ensemble (see for instance \cite{tuckermann2010}). Its distribution is defined by the Boltzmann-Gibbs density, which models the configurations of a conservative system in contact with a heat bath at fixed temperature. Numerically, high-dimensional averages with respect to~$\mu$ are often approximated as ergodic averages over realizations of appropriate stochastic differential equations (SDEs):
\begin{equation}
 \lim_{t\rightarrow\infty}\hat{A}_t=\E_{\mu}(A) \quad \text{a.s.}, \qquad \hat{A}_t:=\frac{1}{t}\int_0^tA(p_s,q_s) \, ds\,.
\label{def convergence ergodic averages}
\end{equation}
A typical dynamics to this end is the Langevin dynamics
\begin{equation}
\left\{
\begin{aligned}
\ds dq_t &= M^{-1} p_t \, dt, \\
\ds dp_t &= -\n V(q_t) \, dt-\gamma M^{-1} p_t \, dt + \sqrt{\frac{2\gamma}{\beta}} \, dW_t\,,
\end{aligned}
\right.
\label{Langevin}
\end{equation}
where $dW_t$ is a standard Wiener process, $V$ the potential energy function, $\gamma > 0$ a friction coefficient, $M$ the mass matrix of the system, and $\beta$ is proportional to the inverse temperature (see Section~\ref{section modified langevin dynamics} for more precise definitions). For references on the ergodicity of Langevin dynamics, we refer the reader to~\cite{Talay} and \cite{mattingly2002ergodicity}, for instance.

There are two main sources of error in the computation of average properties such as $\mathbb{E}_\mu(A)$ through time averages as in~\eqref{def convergence ergodic averages}: (i) a systematic bias (or {\em perfect sampling bias}) related to the use of a discretization method for the SDEs (and usually proportional to a power of the integration step size $\Delta t$), and  (ii) statistical errors, due to the finite lengths of the sampling paths involved and the underlying variance of the random variables. The first point was studied in~\cite{Matthews} for standard Langevin dynamics. Our focus in this work is on the statistical error.

Statistical errors may be large when the dynamics is metastable, \emph{i.e.} when the system remains trapped for a very long time in some region of the configuration space (called a metastable region) before hopping to another metastable region. Metastability implies that the convergence of averages over trajectories is very slow, and that transitions between metastable regions (which are typically the events of interest at the macroscopic level) are very rare. In fact, metastability arises from the multi-modality of the probability measure sampled by the dynamics. We refer for instance to~\cite{lelievre-13} for a review on ways to quantify the metastability of sampling dynamics. There are various strategies to reduce the variance of time averages by reducing the metastability. The most famous one is importance sampling: the potential energy function~$V$ is modified by an additional term~$\widetilde{V}$ so that the Langevin dynamics associated with $V+\widetilde{V}$ is less metastable. An automatic way of doing so is to consider a so-called reaction coordinate, and define $\widetilde{V}$ as the opposite of the associated free energy (see \cite{lelievre2010free,LS15} for further precisions).

We explore here an alternative path, which consists in modifying the kinetic energy rather than the potential energy. Indeed, recall that the difficult part in sampling the canonical measure is in sampling positions (see Section~\ref{section modified langevin dynamics} for a more precise discussion of this point). There is therefore some freedom in the choice of the kinetic energy if the goal is to compute average properties.

Previous works in this direction focused on changing the mass matrix in order to increase the time steps used in the simulation (see \emph{e.g.}~\cite{bennett1975mass,plechac2010implicit}). The mathematical analysis we provide is inspired by a recent work by~\cite{PRL-ARPS} where the kinetic energy of each particle is more drastically modified: it is set to~0 when the particle's momenta are small, while it remains unchanged for larger momenta. In such \emph{adaptively restrained} (AR) simulations, particles may become temporarily \emph{frozen}, while their momenta may continue to evolve. Since, in many cases, inter-particle forces only depend on relative particle positions, and hence do not have to be updated when particles are frozen, adaptively restrained particle simulations may yield a significant algorithmic speed-up~$S_{\rm algo}$ when a sufficiently large number of particles are frozen at each time step (or, more generally, when inter-particle distances remain constant and particle forces are expressed in local reference frames). This has been demonstrated in several contexts, \emph{e.g.} for modeling hydrocarbon systems (\cite{Bosson20122581}), proteins (\cite{Rossi01072007}), and for electronic structure calculations (\cite{BossonBAQM}).

Unfortunately, freezing particles even temporarily may make iterates more correlated, which may translate into an increase of the statistical error~$\sigma^2_{\rm mod}$ observed for modified Langevin dynamics, compared to the statistical error~$\sigma_{\rm std}^2$ observed for standard Langevin dynamics. The actual speed-up of the method, in terms of the total wall-clock time needed to achieve a given precision in the estimation of an observable, should therefore be expressed as:
\begin{equation}
S_{\rm actual} = S_{\rm algo} \frac{\sigma^2_{\rm std}}{\sigma^2_{\rm mod}}\,.
\label{cout}
\end{equation}
Our aim here is thus to quantify the increase in the variance as a function of the parameters of the modified kinetic energy. In fact, a first task is to prove that the Langevin dynamics with modified kinetic energy is indeed ergodic, and that the variance is well defined. This is unclear at first sight since the modified dynamics fails to be hypoelliptic (see the discussion in Section~\ref{sec:cv_ergo_avg}).

\medskip

This article is organized as follows. In Section~\ref{section modified langevin dynamics}, we introduce the modified Langevin dynamics we consider, and present the particular case of the AR-Langevin dynamic. The ergodicity of these dynamics is proved in Section~\ref{section ergodicity}, both in terms of almost-sure convergence of time averages along a single realization, and in terms of the law of the process. We also provide a result on the regularity of the evolution semi-group, adapted from similar estimates for standard Langevin dynamics in ~\cite{Talay}. Such estimates allow us to analyze the statistical error in Section~\ref{section analysis of statistical error}. We state in particular a Central Limit Theorem for~$\widehat{A}_t$, and perform a perturbative study of the asymptotic variance of the AR-Langevin dynamics in some limiting regime. Our theoretical findings are illustrated by numerical simulations in Section~\ref{sec:numerics}, both in a simple one-dimensional case where the variance can be accurately computed using an appropriate Galerkin approximation, as well as for a more realistic system for which we resort to Monte-Carlo simulations. The proofs of our results are gathered in Section~\ref{section proofs}.

%--------------------------------------------------------------------
\section{Modified Langevin dynamics}
\label{section modified langevin dynamics}

We consider a system of $N$ particles in spatial dimension~$\dspace$, so that the total dimension of the system is $d:=\dspace\times N$. The vectors of positions and momenta are denoted respectively by $q=(q_1, \cdots, q_N)$ and $p=(p_1, \cdots, p_N)$. Periodic boundary conditions are used for positions, so that the phase-space of admissible configurations is $\mathcal{E} = \mathcal{D}\times \R^{d}$ with $\mathcal{D}:=\left(L\mathbb{T}\right)^{d}$, $\mathbb{T}=\R\backslash \mathbb{Z}$ being the one-dimensional unit torus and $L>0$ the size of the simulation box.

In order to possibly increase the rate of convergence of the ergodic averages~\eqref{def convergence ergodic averages}, we modify the Langevin dynamics~\eqref{Langevin} by changing the kinetic energy. More precisely, instead of the standard quadratic kinetic energy
\[
U_{\mr{std}}(p)=\frac{1}{2}p^TM^{-1}p, \qquad M = \mathrm{diag}(m_1,\dots,m_N),
\]
we introduce a general kinetic energy function $U:\R^{d}\rightarrow \R$. The total energy of the system is then characterized by the Hamiltonian
\begin{equation}
\label{eq: modified Hamiltonian}
H(p,q)=\Um(p)+V(q).
\end{equation}
In order to ensure that the measure $\mathrm{e}^{-\beta H(q,p)}\, dq\,dp$ can be normalized, and in order to simplify the mathematical analysis, we make in the sequel the following assumption.

\begin{assumption}
\label{assumption potential}
The potential energy function $V$ belongs to $C^{\infty}(\mathcal{D}, \R)$, and $U\in C^{\infty}(\R^{d},\R)$ grows sufficiently fast at infinity in order to ensure that $\mr{e}^{-\beta U}\in L^1(\R^d)$.
\end{assumption}

The Langevin dynamics associated with a general Hamiltonian reads
\[
\left\{
\begin{aligned}
dq_t & = \n_p H(p_t, q_t) \, dt, \\
dp_t & = -\n_q H(p_t, q_t) \, dt - \gamma \n_p H(p_t, q_t) \, dt + \sqrt{\frac{2\gamma}{\beta}} \, dW_t,
\end{aligned}
\right.
\]
where $dW_t$ is a standard $d$-dimensional Wiener process and $\gamma>0$ is the friction constant. For the separable Hamiltonian~\eqref{eq: modified Hamiltonian}, the general Langevin dynamics simplifies as
\begin{equation}
  \label{modified Langevin}
  \left\{
  \begin{aligned}
    dq_t & = \n U(p_t) \, dt, \\
    dp_t & = -\n V(q_t) \, dt - \gamma \n U(p_t) \, dt + \sqrt{\frac{2\gamma}{\beta}} \, dW_t.
  \end{aligned}
\right.
\end{equation}
The generator of the process (\ref{modified Langevin}) reads
\begin{equation}
\Lopgen = \n \U \cdot \n_q - \n V\cdot \n_p+\gamma\left(-\n \U\cdot \n_p + \frac{1}{\beta}\de_p\right).
\label{eq: generator modified Langevin}
\end{equation}
A simple computation shows that the canonical distribution
\begin{equation}
\mu(dq \,dp)=Z^{-1}_{\mu}\mathrm{e}^{-\beta H(q,p)}\,dp\,dq, \qquad Z_{\mu}=\int_{\El}\mathrm{e}^{-\beta H(q,p)}\,dp\,dq < +\infty,
\label{eq: invariant measure modified Langevin}
\end{equation}
is invariant under the dynamics~\eqref{modified Langevin}, \emph{i.e.} for all $C^\infty$ functions $\phi$ with compact support,
\[
\int_{\El}\Lopgen\phi \ d\mu=0.
\]

Note that, in view of the separability of the Hamiltonian, the marginal of the distribution $\mu$ in the position variables is, for any kinetic energy~$U$,
\[
\bar{\mu}(dq) = Z_{V}^{-1} \mr{e}^{-\beta V(q)} \, dq, \qquad Z_V = \int_{\mathcal{D}} \mr{e}^{-\beta V(q)} \, dq.
\]
In particular, this marginal distribution therefore coincides with the one of the standard Langevin dynamics~\eqref{Langevin}. This allows to straightforwardly estimate canonical averages of observables depending only on the positions with the modified Langevin dynamics~\eqref{modified Langevin}. In fact, there is no restriction in generality in considering observables depending only on the positions, since general observables $A(q,p)$ depending both on momenta and positions can be reduced to functions of the positions only by a partial integration in the momenta variables. This partial integration is often very easy to perform since momenta are independent Gaussian random variables under the canonical measure associated with the standard kinetic energy.

\subsection{AR-Langevin dynamics}
\label{sec:ARPS_presentation}

A concrete example for the choice of the kinetic energy function $U$ in (\ref{eq: modified Hamiltonian}) is the one proposed for the adaptively restrained Langevin dynamics in~\cite{PRL-ARPS}. It is parameterized by two constants $0\leq \er<\ef$. In this model, the kinetic energy is a sum of individual contributions
\[
U(p)=\sum_{i=1}^N u(p_i).
\]
For large values of momenta, the modified individual kinetic energies are equal to the standard kinetic energy of one particle, but they vanish for small momenta:
\[
u(p_i) = \left\{
\begin{aligned}
  \ds 0                 &\quad \mathrm{for} \ \frac{p_i^2}{2m_i} \leq \er, \\
  \ds \frac{p_i^2}{2m_i} &\quad \mathrm{for} \ \frac{p_i^2}{2m_i} \geq \ef.
\end{aligned}
\right.
\]
An appropriate function allows to smoothly interpolate between these two limiting regimes (see Definition~\ref{definition ARPS Langevin} for the precise expression). A possible choice of an individual kinetic energy~$u$, as well as the associated canonical distribution of momenta $Z_{u}^{-1} \mathrm{e}^{-\beta u(p)} \, dp$ are depicted in Figure~\ref{figure 1 left} and Figure~\ref{figure 1 right} when $D=1$.

\begin{figure}[t]
  %\vspace{-100pt}
  \centering
  \includegraphics[width=0.8\textwidth]{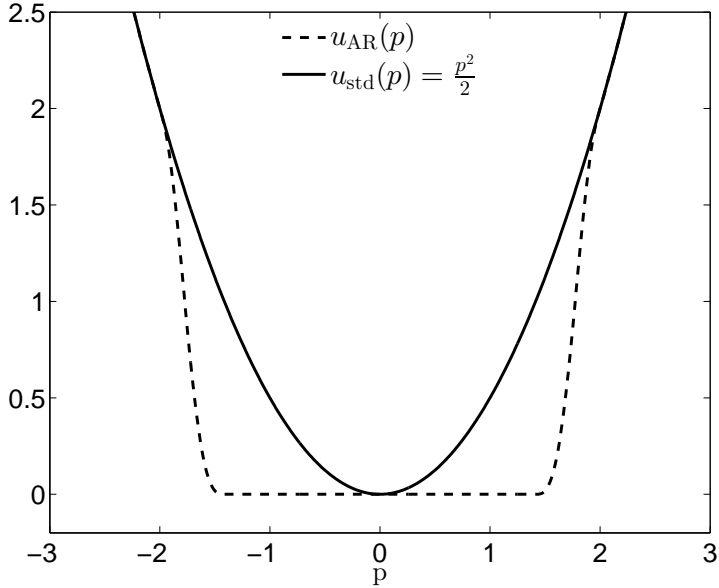}
  %\vspace{100pt}
  \caption{Standard quadratic kinetic energy function $\Ustd$ (solid lines), and an example of an AR kinetic energy function~$u$ with parameters $\ef=2$ and $\er=1$ (dashed line).}
  \label{figure 1 left}
	\end{figure}

\begin{figure}[t]
  %\vspace{-100pt}
  \centering
  \includegraphics[width=0.8\textwidth]{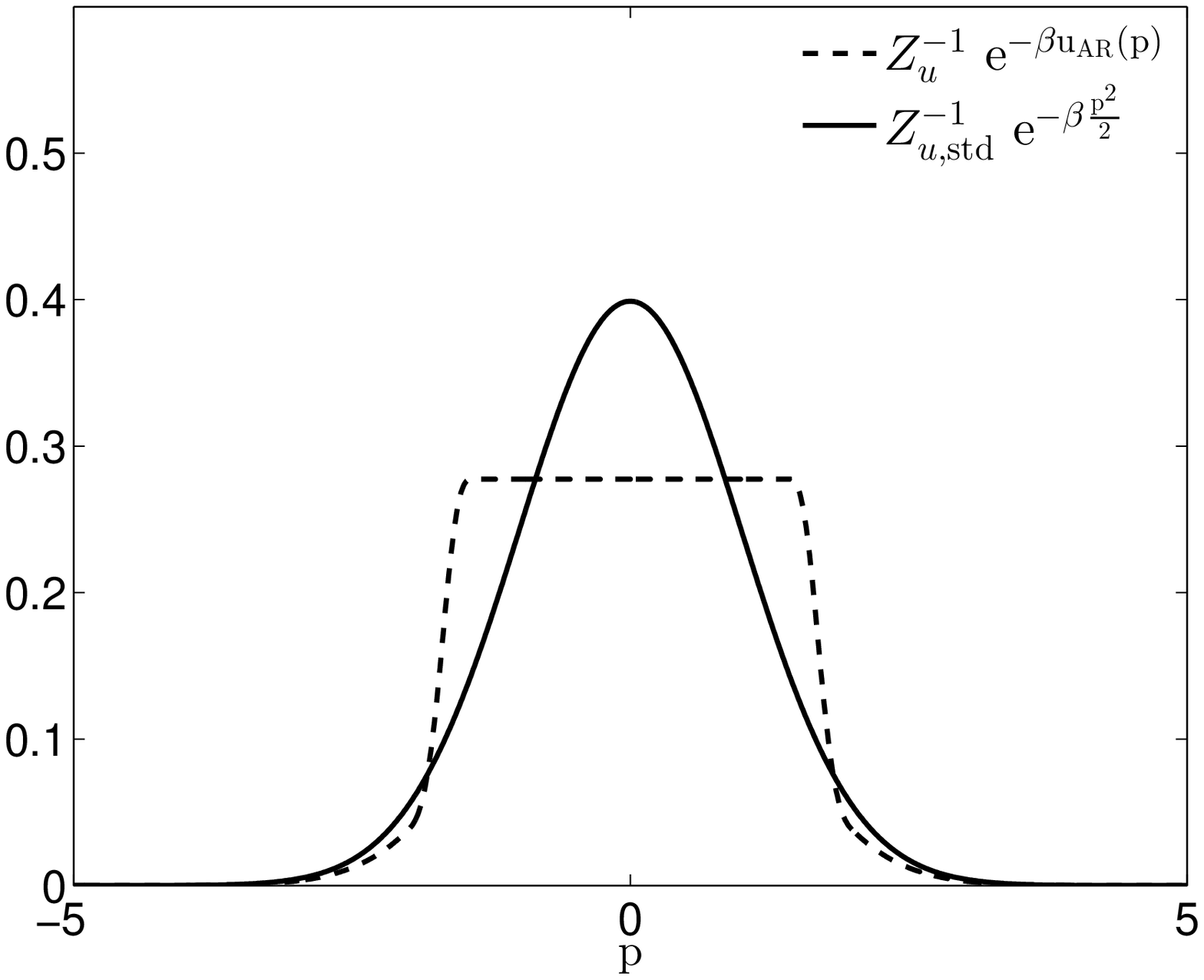}
  %\vspace{100pt}
  \caption{Marginal canonical densities associated with the kinetic energy functions of Figure \ref{figure 1 left}.}
  \label{figure 1 right}
	\end{figure}

%%%%% OLD FIGURE FORMAT
%%%\begin{figure}[t]
  %%%\centering
  %%%%%\vspace{-100pt}
  %%%\begin{subfigure}[b]{0.5\textwidth}
    %%%\includegraphics[width=\textwidth]{Figures/modifiedHamiltonianIntroduction_ui.eps}%{Figures/comparisonFzeroAndfepsilonK.png}
  %%%\end{subfigure}%
  %%%~ %add desired spacing between images, e. g. ~, \quad, \qquad, \hfill etc.
  %%%%(or a blank line to force the subfigure onto a new line)
  %%%\begin{subfigure}[b]{0.5\textwidth}
    %%%\includegraphics[width=\textwidth]{Figures/distributionFunctionIntroduction_ui.eps}
  %%%\end{subfigure}
  %%%\vspace{80pt}
  %%%\caption{Left: Standard quadratic kinetic energy function $\Ustd$ (solid lines), and an example of an AR kinetic energy function~$u$ with parameters $\ef=2$ and $\er=1$ (dashed line). Right: Associated marginal canonical densities.}%\label{fig:comparisonFzeroAndfepsilonK}
  %%%\label{figure 1}
%%%\end{figure}

The interest of AR-Langevin dynamics is that, when their individual kinetic energies are sufficiently small, particles do not move. When two particles are frozen in this way, their pairwise interactions need not be updated. This allows decreasing the computational complexity of the force computation, which is typically the most time-consuming part of a molecular dynamics solver. Note that this can be generalized to higher-order interactions (such as three-body interactions based on bending angles for instance).

\begin{remark}
\label{remark individual parameters}
Note that, due to the additive structure of the kinetic energy, the momenta~$p_i$ are independent and identically distributed (i.i.d.) under the canonical measure. It is however possible to choose different parameters $\er$ and $\ef$ for different particles, for example to focus calculations on a specific part of the particle system, in which case the momenta are still independent but not longer identically distributed. Such a situation is considered in the numerical example presented in Section~\ref{subsection A more realistic system: second order discretization}.
\end{remark}

\section{Ergodicity of the modified Langevin dynamics}
\label{section ergodicity}

There are several notions of ergodicity for stochastic processes. We focus here on two of them: the convergence of ergodic averages over a single trajectory, and the convergence of the law of the process.

\subsection{Convergence of ergodic averages}
\label{sec:cv_ergo_avg}

The convergence of ergodic averages over one trajectory is automatically ensured by the existence of an invariant probability measure and the irreducibility of the dynamics (see for instance~\cite{Kli87,MT93} for early results on such convergences for possibly degenerate diffusions). Since, by construction, an invariant probability measure is known (namely the canonical measure~\eqref{eq: invariant measure modified Langevin}), it suffices to show that the process generated by the modified Langevin equation is irreducible to conclude to the convergence of ergodic averages.

As reviewed in~\cite{rey-bellet}, the most standard argument to prove the irreducibility of degenerate diffusions is to prove the controllability of the dynamics relying on the Stroock-Varadhan support theorem, and the regularity of the transition kernel thanks to some hypoellipticity property. These conditions are satisfied for standard Langevin dynamics (see for instance~\cite{mattingly2002ergodicity}), but not for the modified Langevin dynamics we consider, since the Hessian of the kinetic energy function may not be invertible on an open set. This is the case for the AR kinetic energy function presented in Section~\ref{sec:ARPS_presentation}.

To illustrate this point, let us show for instance how the standard way of proving hypoellipticity fails (the proof of the controllability faces similar issues). The first task is to rewrite the generator~\eqref{eq: generator modified Langevin} of the process as
\[
\Lopgen = X_0 - \sum_{j=1}^{d} X_j^{\dagger}X_j,
\]
where
\[
X_0 = \n \U \cdot \n_q - \n V\cdot \n_p-\gamma\n \U\cdot \n_p,
\qquad
X_j=\sqrt{\frac{\gamma}{\beta}} \, \du_{p_j},
\]
and $X_j^{\dagger}$ is the adjoint of $X_j$ on the flat space $L^2(\El)$. We next compute, for $j = 1, \dots, d$, the commutators
\[
%\begin{equation}
%\label{eq: commutator X0Xj}
\ds\left[X_0,X_j\right] = X_0X_j - X_j X_0 = \sqrt{\frac{\gamma}{\beta}}\,\n\left(\du_{p_j} U\right)\cdot \left(\n_q-\gamma \n_p\right).
%\end{equation}
\]
When $\nabla^2 U$ is invertible, it is possible to recover the full algebra of derivatives by an appropriate combination of $X_1,\dots,X_d$ and $[X_0,X_1],\dots,[X_0,X_d]$. Here, we consider a situation when this is not the case and, even more dramatically, where the Hessian may vanish on an open set. In this situation, $\left[X_0,X_j\right]=0$ on the same open set, and in fact all iterated commutators $[X_0,[\dots[X_0,X_j]]]$ also vanish.

We solve this problem by a direct constructive approach, where we see the modified dynamics as a perturbation of the standard Langevin dynamics. We rely on the following assumption:

\begin{assumption}
The kinetic energy function $U\in C^{\infty}$ of the modified Langevin dynamics is such that
\[
\left\|\n\U - \n\U_{\rm std}\right\|_{L^{\infty}}\leq \ConstNablaU
\]
for some constant $\ConstNablaU < +\infty$.
\label{hyp mod Langevin ergodicity}
\end{assumption}

Under this assumption, we can prove that the modified Langevin dynamics is irreducible by proving an appropriate minorization condition, which crucially relies on the compactness of the position space~$\mathcal{D}$ (see Section~\ref{proof of the minorization condition} for the proof).

\begin{lemma}[Minorization condition]
\label{lemma minorization condition}
Suppose Assumption~\ref{hyp mod Langevin ergodicity} holds. Then for any fixed $p^*>0$ and $t>0$, there exists a probability measure $\nu_{p^*, t}$ on $\mathcal{D}\times \R^d$ and a constant $\kappa>0$ such that, for every Borel set $B\in \mathscr{B}(\mathcal{E})$,
\[
\mathbb{P}\left((q_{t},p_{t})\in B \, \Big| \, \left|p_{0}\right| \leq p_{*} \right) \geq \kappa \, \nu_{p_{*},t}(B),
\]
with $ \nu_{p_{*},t}(B)>0$ when $\left|B\right|>0$.
\end{lemma}

The minorization condition implies the irreducibility of the dynamics, so that the following convergence result readily follows.

\begin{theorem}[Convergence of ergodic averages]
\label{thm convergence of averages}
When Assumption~\ref{hyp mod Langevin ergodicity} holds, ergodic averages over trajectories almost surely converge to the canonical average:
\[
\forall A\in L^{1}(\mu), \qquad \lim_{t \rightarrow +\infty}\frac{1}{t}\int_0^t A(q_s,p_s)ds = \int_{\El}A \, d\mu \quad \mathrm{a.s.}
\]
\end{theorem}

\subsection{Convergence of the law}

There are various functional frameworks to measure the convergence of the law of the process. We consider here weighted $L^\infty$ estimates on the semi-group $\mathrm{e}^{t \Lope}$. More precisely, we introduce a scale of Lyapunov functions
\begin{equation}
\K_s(p):=1+\left|p\right|^{2s}
\label{eq: Lyapunov function}
\end{equation}
 for $s\in\mathbb{N}^{*}$. Recall indeed that only momenta need to be controlled since positions remain in a compact space. The associated weighted $L^\infty$ spaces are
\[
L^{\infty}_{\K_s} = \left\{ f \, \mathrm{measurable} \, \left| \, \left\|f\right\|_{L^{\infty}_{\K_s}}:=\left\|\frac{f}{\K_s}\right\|_{L^{\infty}} < +\infty \right. \right\}.
\]
In order to prove the exponential convergence of the law, we rely on the result of~\cite{HairerMattingly-Yet}, which states that if a Lyapunov condition and a minorization condition hold true, then the sampled chain converges exponentially fast to its steady state in the following sense.

\begin{theorem}[Exponential convergence of the law]
\label{thm uniform convergence}
Suppose that Assumption~\ref{hyp mod Langevin ergodicity} holds. Then the invariant measure $\mu$ is unique, and for any $s\in \mathbb{N}^*$, there exist constants $C_s, \lambda_s>0$ such that
\begin{equation}
\label{eq thm uniform convergence}
\forall f \in L^\infty_{\K_s}, \quad \forall \ t\geq 0, \qquad \left\|\mathrm{e}^{t\Lope}f-\int_{\mathcal{E}}f d\mu\right\|_{L^{\infty}_{\K_{s}}}\leq C_s \mathrm{e}^{-\lambda_s t}\left\|f\right\|_{L^{\infty}_{\K_{s}}}.
\end{equation}
\end{theorem}

As mentioned above, the proof of this result directly follows from the arguments of~\cite{HairerMattingly-Yet}. The minorization condition is already stated in Lemma~\ref{lemma minorization condition}, while the appropriate Lyapunov condition reads as follows (see Section~\ref{proof of the lyapunov condition} for the proof, which uses the same strategy as~\cite{Matthews} and~\cite{Joubaud}).

\begin{lemma}[Lyapunov Condition]
  \label{lemma Lyapunov}
  Suppose that Assumption~\ref{hyp mod Langevin ergodicity} holds. Then, for any $s\geq 1$ and $t>0$, there exist $b>0$ and $a\in\left[0,1\right)$ such that
    \[
      \mathrm{e}^{t\Lope}\K_s \leq a\K_s + b.
      %\label{Lyapunov}
    %\end{equation}
		\]
\end{lemma}

%------------------------ Talay/Kopec ----------------------
\subsection{Regularity results for the evolution semi-group}

We provide in this section decay estimates for the spatial derivatives of $\mathrm{e}^{t \Lope}f$, following the approach pioneered in~\cite{Talay} and further refined in~\cite{Kopec}. Such estimates were obtained for the standard Langevin dynamics, but can in fact straightforwardly be extended to modified Langevin dynamics with Hessians bounded from below by a positive constant. Our aim in this section is to provide decay estimates for the spatial derivatives of $\mathrm{e}^{t \Lope}f$ in the situation when $\nabla^2 U$ fails to be strictly convex, for instance because $\nabla^2 U$ vanishes on an open set as is the case for AR particle simulations.

In order to state our results, we first need to define the weighted Sobolev spaces $\mathscr{W}^{n,\infty}_{\K_{s}}$ for $n\in \mathbb{N}$:
\[
\ds\mathscr{W}^{n,\infty}_{\K_{s}}=\Big\{f\in L^{\infty}_{\K_{s}} \left| \ \forall k\in \mathbb{N}^{2d}, \ \left|k\right|\leq n, \ \du^k f\in L^{\infty}_{\K_{s}}\Big\}\right. .
\]
These spaces gather all functions which grow at most like $\K_s$, and whose derivatives of order at most~$n$ all grow at most like $\K_s$. We also introduce the space of smooth functions $\mathscr{S}$, the vector space of functions $\ds f\in L^2\left(\mu\right)$ such that, for any $n \geq 0$, there exists $r \in \mathbb{N}$ for which $f\in \mathscr{W}^{n,\infty}_{\K_{r}}$.

We also make the following assumption on the kinetic energy function, which can be understood as a condition of ``almost strict convexity'' of the Hessian $\n^2U$.

\begin{assumption}
\label{hyp THM}
The kinetic energy $U\in \mathscr{S}$ has bounded second-order derivatives:
\begin{equation}
\sup_{\left|j\right|=2}\left\|\du^j U\right\|_{L^{\infty}}<\infty\,,
\label{eq: assummption Linfty du2}
\end{equation}
and there exist a function $U_{\nu} \in \mathscr{S}$ and constants $\nu>0$ and $G_{\nu}\geq 0$ such that
\begin{equation}
\n^2 U_{\nu}\geq \nu >0
\label{condition dn2U strictly pos}
\end{equation}
and
\begin{equation}
\left\|\n \left(U-U_{\nu}\right)\right\|_{L^{\infty}}\leq G_{\nu}.
\label{condition dnU-dnUnu}
\end{equation}
\end{assumption}

\begin{remark}
A natural choice for the function $U_{\nu}$ in Assumption \ref{hyp THM} is $\Ustd$. The condition~\eqref{condition dn2U strictly pos} then holds with $\nu = 1/\max(m_1,\dots,m_N)$. Moreover, \eqref{eq: assummption Linfty du2} holds true as soon as $U$ is a local perturbation of $U_{\rm std}$. The most demanding condition therefore is~\eqref{condition dnU-dnUnu}, especially if $G_\nu$ has to be small.
\end{remark}

By following the same strategy as in~\cite[Proposition A.1.]{Kopec} (which refines the results already obtained in~\cite{Talay}), and appropriately taking care of the lack of strict positivity of the Hessian $\nabla^2 U$ by assuming that $G_\nu$ is sufficiently small, we prove the following result in Section~\ref{proof proposition kopec}.

\begin{lemma}
\label{main theorem kopec}
Suppose that Assumptions~\ref{hyp mod Langevin ergodicity} and~\ref{hyp THM} hold, and fix $A\in\mathscr{S}$. For any $n\geq 1$, there exist $\widetilde{n},s_n\in \mathbb{N}$ and $\lambda_n>0$ such that, for $s \geq s_n$ and $\displaystyle G_{\nu}\leq\rho_{s}$ with $\rho_{s}>0$ sufficiently small (depending on~$s$ but not on~$n$), there is $r \in \mathbb{N}$ and $C>0$ for which
\begin{equation}
\forall t\geq 0, \quad \forall \left|k\right|\leq n, \qquad \left\|\du^k \mathrm{e}^{t\Lopgen}\Pi_{\mu}A\right\|_{L^{\infty}_{\K_s}}\leq C  \left\| A\right\|_{W^{\widetilde{n},\infty}_{\K_r}}\mathrm{e}^{-\lambda_n t}.
\label{eq: functional estimates Kopec}
\end{equation}
\end{lemma}

The parameter~$\rho_s$ can in fact be made explicit, see~\eqref{eq: def rho m} below. The decay estimate~\eqref{eq: functional estimates Kopec} shows that the derivatives of the evolution operator can be controlled in appropriate weighted Hilbert spaces. Note however that the Lyapunov functions entering in the estimates are not the same a priori on both sides of the inequality~\eqref{eq: functional estimates Kopec}. Let us emphasize, though, that we can obtain a control in all spaces $L^\infty_{\K_s}$ for $s$ sufficiently large (depending on the order of derivation).

%-------------------------- STATISTICAL ERROR -------------
\section{Analysis of the statistical error}
\label{section analysis of statistical error}

The asymptotic variance characterizes the statistical error. In Section~\ref{sec:asymptotic_variance}, we show that the asymptotic variance is well defined for the modified Langevin dynamics. We can in fact prove a stronger result, namely that a Central Limit Theorem (CLT) holds true for ergodic averages over one trajectory. In a second step, we more carefully analyze in Section~\ref{sec:perturbative_variance} the properties of the variance of the AR-Langevin dynamics by proving a linear response result in the limit of a vanishing lower bound on the kinetic energies. To obtain the latter results, we rely on the estimates provided by Lemma~\ref{main theorem kopec}.

\subsection{A Central Limit theorem for ergodic averages}
\label{sec:asymptotic_variance}

Let us first write the asymptotic variance in terms of the generator of the dynamics. To simplify the notation, we introduce the orthogonal projection $\Pi_{\mu}$ onto the orthogonal of the kernel of the operator~$\Lopgen$ (with respect to the $L^2(\mu)$ scalar product): for any $\psi \in L^2(\mu)$,
%\begin{equation}
%\label{def projection}
\[
\Pi_{\mu} \psi:=\psi - \int_{\mathcal{E}}\psi \, d\mu.
\]
%\end{equation}
Since $L^{\infty}_{\K_{s}} \subset L^2(\mu)$, we can define $\ds\widetilde{L^{\infty}_{\K_{s}}}=\Pi_{\mu}\left( L^{\infty}_{\K_{s}}\right)$. The ergodicity result~\eqref{eq thm uniform convergence} allows us to conclude that the operator $\Lopgen$ is invertible on~$\widetilde{L^{\infty}_{\K_{s}}}$ since the following operator equality holds on $\mathcal{B}\left(\widetilde{L^{\infty}_{\K_{s}}}\right)$, the Banach space of bounded operators on~$\widetilde{L^{\infty}_{\K_{s}}}$:
\[
\Lopgen^{-1} = \int_0^{+\infty} \mathrm{e}^{t \Lopgen} \, dt.
\]
This leads to the following resolvent bounds (the second part being a direct corollary of Lemma~\ref{main theorem kopec}).

\begin{corollary}
\label{corollary convergence thm}
Suppose that Assumption~\ref{hyp mod Langevin ergodicity} holds. Then, for any $s \in \mathbb{N}^*$,
\begin{equation}
\label{eq: corollary convergence thm}
\left\|\Lopgen^{-1}\right\|_{\mathcal{B}\left(\widetilde{L^{\infty}_{\K_{s}}}\right)}\leq \frac{C_s}{\lambda_s},
\end{equation}
where $\lambda_s,C_s$ are the constants introduced in Theorem~\ref{thm uniform convergence}.
%% DERIVATIVES
Suppose in addition that Assumption~\ref{hyp THM} holds, and fix $A\in\mathscr{S}$. For any $n\geq 1$, there exist $\widetilde{n},s_n\in \mathbb{N}$ and $\lambda_n>0$ such that, for $s \geq s_n$ and $\displaystyle G_{\nu}\leq\rho_{s}$ with $\rho_{s}>0$ sufficiently small (depending on~$s$ but not on~$n$), there is $r \in \mathbb{N}$ and $C>0$ for which
\begin{equation}
\label{eq: functional estimates Kopec_resolvent}
\forall \left|k\right|\leq n, \qquad \left\|\du^k \Lopgen^{-1} \Pi_{\mu}A\right\|_{L^{\infty}_{\K_s}}\leq \frac{C}{\lambda_n}  \left\| A\right\|_{W^{\widetilde{n},\infty}_{\K_r}}.
\end{equation}
\end{corollary}

This already allows us to conclude that the asymptotic variance of the time average $\widehat{A}_t$ defined in~\eqref{def convergence ergodic averages} is well defined for any observable $A\in L^{\infty}_{\K_{r}}$ since
\[
\begin{aligned}
\sigma_A^2 &= \lim_{t \rightarrow \infty} t \, \E\Big[ \left(\hat{A}_t-\E_{\mu}\left(A\right)\right)^2\Big] \\
 & = \lim_{t \rightarrow \infty}2\int_0^t\int_{\El}\left(1-\frac{s}{t}\right)\left(\rm{e}^{s\Lopgen}\Pi_{\mu}A\right)\left(\Pi_{\mu}A \right) d\mu\\
&= 2\int_0^{\infty}\int_{\El}\left(\mr{e}^{s\Lopgen}\Pi_{\mu}A\right)\left(\Pi_{\mu}A\right) d\mu
\end{aligned}
\]
by the dominated convergence theorem. Therefore,
\begin{equation}
\label{eq: variance expression}
\sigma^{2}_{A}=2\int_{\El} \left(\Pi_{\mu} A\right) \left(-\Lope^{-1}\Pi_{\mu}A\right) \, d\mu.
\end{equation}
In fact, a Central Limit Theorem can be shown to hold for $\widehat{A}_t$ using standard results (see \emph{e.g.}~\cite{bhattacharya1982functional}).

%-------------------- perturbation ARPS ------------------------
\subsection{Perturbative study of the variance for the AR-Langevin dynamics}
\label{sec:perturbative_variance}

Our aim in this section is to better understand, from a quantitative viewpoint, the behavior of the asymptotic variance for the AR-Langevin dynamics defined in Section~\ref{sec:ARPS_presentation}, at least in some limiting regime where the parameter $\er$ is small. For intermediate values, we need to rely on numerical simulations (see Section~\ref{sec:numerics}).

The regime where both $\er$ and $\ef$ go to~0 is somewhat singular since the transition from $U(p)=0$ to $U(p)=U_{\rm std}(p)$ becomes quite abrupt, which prevents a rigorous theoretical analysis. The regimes where either $\er$ or $\ef$ go to infinity are also of dubious interest since the dynamics strongly perturbs the standard Langevin dynamics. Therefore, we restrict ourselves to the situation where $\er \to 0$ with $\ef$ fixed.

In order to highlight the dependence of the AR kinetic energy function on the restraining parameters $0\leq \er<\ef$, we denote it by $U_{\er,\ef}$ in the remainder of this section. Let us however first give a more precise definition of this function, having in mind that $\ef$ is fixed while $\er$ eventually goes to~0. We introduce to this end an interpolation function $f_{0,\ef} \in C^{\infty}\left(\R\right)$ such that
\begin{equation}
0 \leq f_{0,\ef} \leq 1, \qquad f_{0,\ef}(x) = 1 \ \mathrm{for} \ x \leq 0, \qquad f_{0,\ef}(x) = 0 \ \mathrm{for} \ x \geq \ef,
\label{eq: def fZeroK 1}
\end{equation}
and
\[
\forall n \geq 1, \qquad f_{0,\ef}^{(n)}(0) = f_{0,\ef}^{(n)}(\ef) = 0.
\label{eq: def fZeroK 2}
\]
We next define an interpolation function $f_{\er,\ef}$ obtained from the function $f_{0,\ef}$ by an appropriate shift of the lower bound and a rescaling. More precisely, $f_{\er,\ef}(x) = f_{0,\ef}(\theta_{\er}(x))$ with
\begin{equation}
\label{eq: definition theta}
\theta_{\er}(x):=\left\{\ds
\begin{aligned}%{cc}\ds
x-\er, & \quad \text{for } x \leq \er, \\
\frac{\ef}{\ef-\er}(x-\er), & \quad \text{for }  \er\leq x\leq \ef, \\
x , & \quad \text{for } x \geq \ef.
\end{aligned}
\right.
\end{equation}
A plot of $f_{\er, \ef}$ is provided in Figure~\ref{figure arps redefinition shift functions left}.

\begin{figure}[t]
  %\vspace{-100pt}
  \centering
  \includegraphics[width=0.8\textwidth]{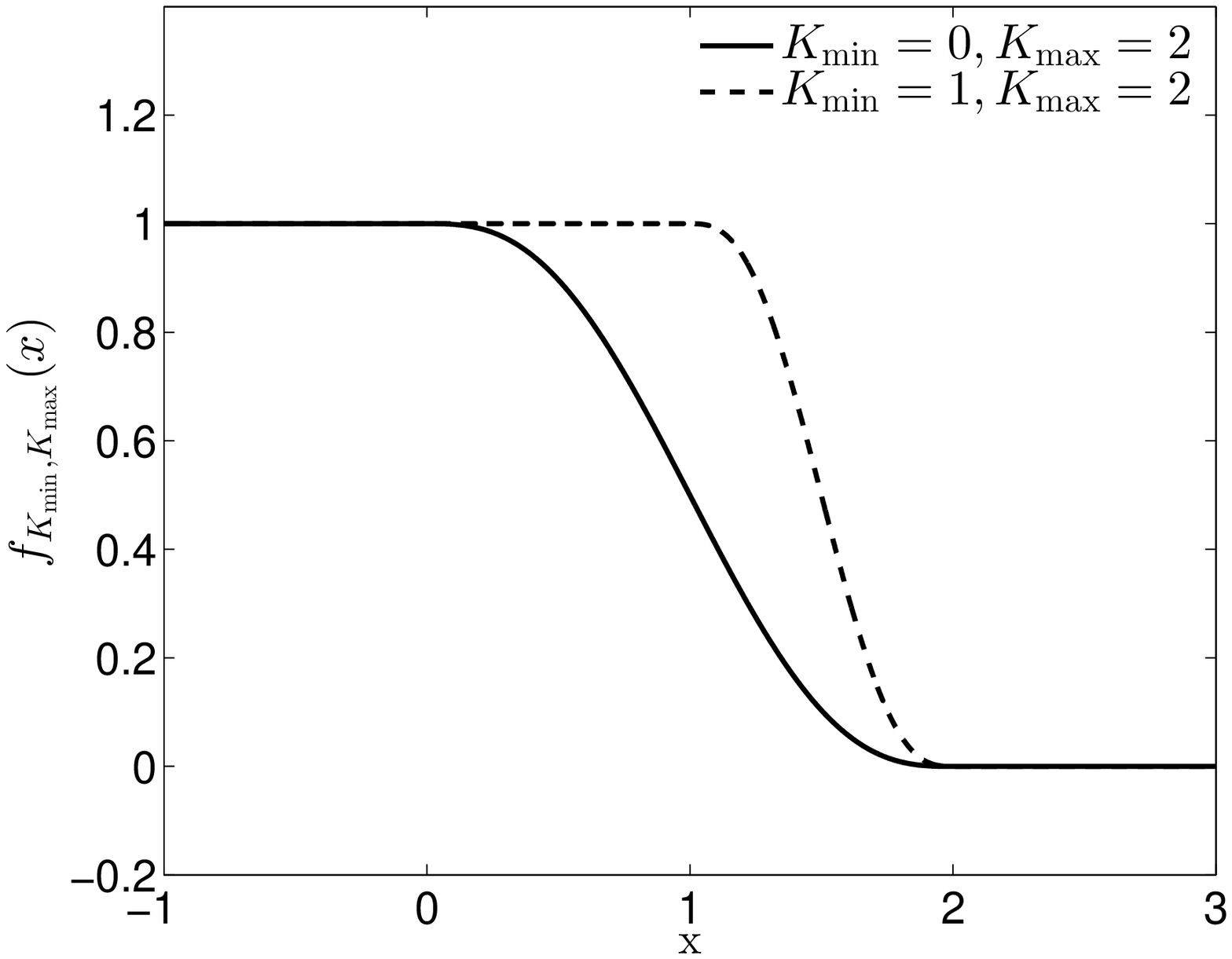}
  %\vspace{100pt}
  \caption{Functions $f_{0,\ef}$ and $f_{\er,\ef}$ for $\ef=2$ and $\er=1$.}
  \label{figure arps redefinition shift functions left}
	\end{figure}

\begin{figure}[t]
  %\vspace{-100pt}
  \centering
  \includegraphics[width=0.8\textwidth]{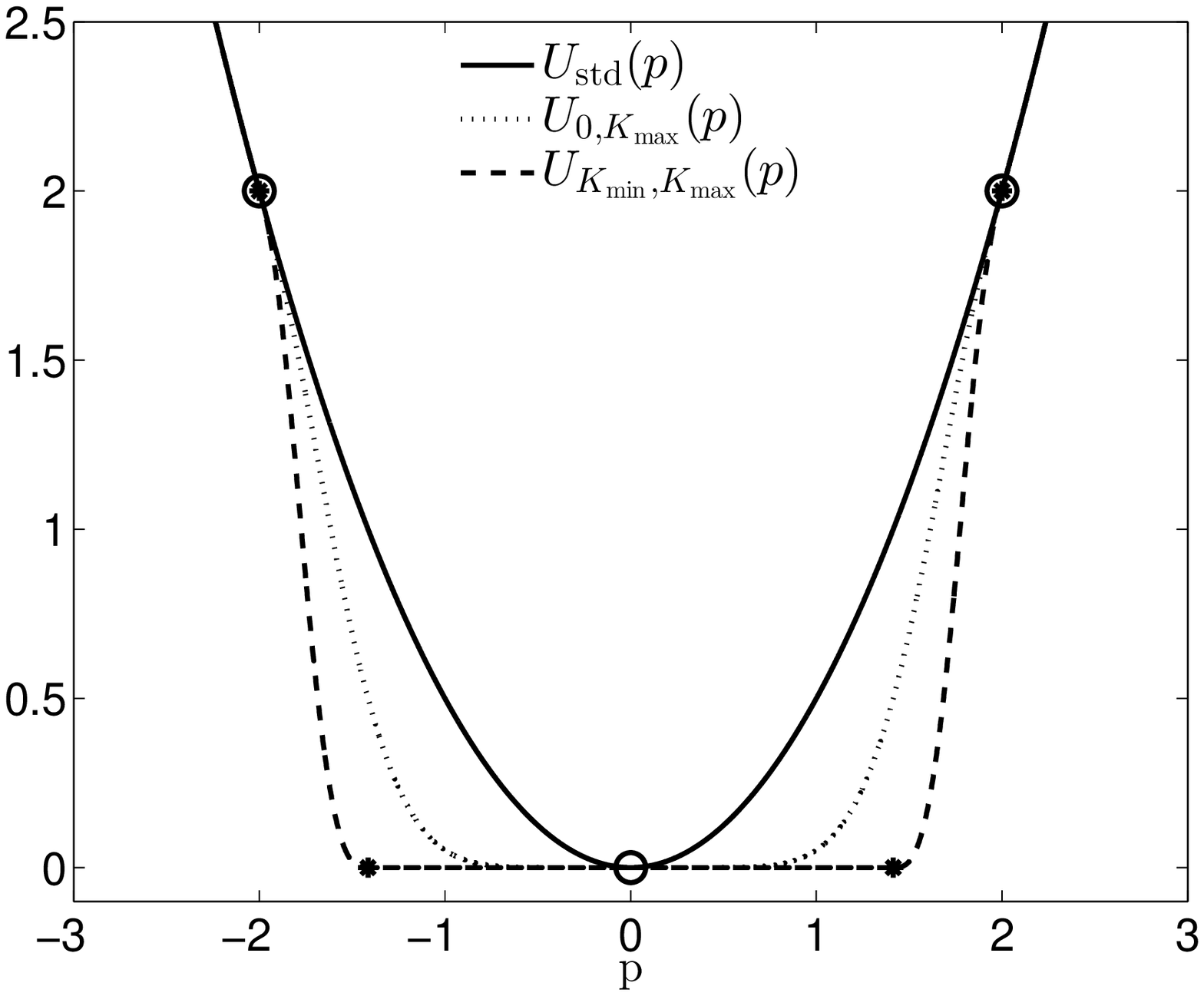}
 % \vspace{100pt}
  \caption{Standard kinetic energy function $\Ustd$, as well as two AR  kinetic energy functions $U_{\er,\ef}$ with $\ef=2$ and $\er=0$ or $1$.}
  \label{figure arps redefinition shift functions right}
	\end{figure}

%%%% OLD FIGURES FORMAT
%%%\begin{figure}[t]
  %%%\centering
  %%%%%\vspace{-100pt}
  %%%\begin{subfigure}[b]{0.5\textwidth}
    %%%\includegraphics[width=\textwidth]{Figures/fSpline.eps}
  %%%\end{subfigure}%
  %%%~ %add desired spacing between images, e. g. ~, \quad, \qquad, \hfill etc.
  %%%%(or a blank line to force the subfigure onto a new line)
  %%%\begin{subfigure}[b]{0.5\textwidth}
    %%%\includegraphics[width=\textwidth]{Figures/modifiedHamiltonian_parametrization.eps}
  %%%\end{subfigure}
  %%%\vspace{50pt}
  %%%\caption{Left: Functions $f_{0,\ef}$ and $f_{\er,\ef}$ for $\ef=2$ and $\er=1$. Right: Standard kinetic energy function $\Ustd$, as well as two AR  kinetic energy functions $U_{\er,\ef}$ with $\ef=2$ and $\er=0$ or $1$. }
  %%%%\label{fig:comparisonFzeroAndfepsilonK}
  %%%\label{figure arps redefinition shift functions}
%%%\end{figure}

\begin{definition}[AR kinetic energy function]
\label{definition ARPS Langevin}
For two parameters $0\leq\er<\ef$, the AR kinetic energy function $U_{\er, \ef}$ is defined as
\begin{equation}
U_{\er, \ef}(p):=\sum_{i=1}^Nu_{\er,\ef}(p_i),
\label{definition U}
\end{equation}
where the individual kinetic energy functions are
\begin{equation}
  \label{eq: definition zeta}
  u_{\er,\ef}(p_i):=\left\{
  \begin{aligned}%{cc}
    \ds
    0,   &\quad \text{for } \frac{p_i^2}{2m_i}\leq \er, \\
    \left[ 1-f_{\er,\ef}\left(\frac{p_i^2}{2m_i}\right)\right]\frac{p_i^2}{2m_i}, &\quad  \text{for } \frac{p_i^2}{2m_i} \in [\er,\ef], \\
    \frac{p_i^2}{2m_i},   &\quad \text{for } \frac{p_i^2}{2m_i}\geq \ef. \\
  \end{aligned}
\right.
\end{equation}
\end{definition}
%Note that
%\[
%\left\|\n \left(U_{\er}-U_{\mr{std}}\right)\right\|_{L^{\infty}}\leq \sqrt{2\ef}+\frac{\sqrt{2\ef^3}}{\ef-\er}\left\|f^{\prime}_{0,\ef}\right\|_{L^{\infty}}\,.
%\]

Of course, $U_{\er,\ef}(p)$ converges to $U_{0,\ef}(p)$ as $\er \to 0$. The limiting kinetic energy function $U_{0,\ef}$ corresponds to what we call the Zero-$\ef$-AR-Langevin dynamics (see Figure~\ref{figure arps redefinition shift functions right} for an illustration). Let us emphasize that the limiting dynamics is not the standard Langevin dynamics, so that the expansion in powers of $\er$ of the variance we provide is with respect to the limiting variance of the dynamics corresponding to~$U_{0,\ef}$. To simplify the notation, we denote by $\sigma^2(\er)$ the variance associated with the kinetic energy~$U_{\er,\ef}$.
%The following proposition is the main result of this section.
\begin{proposition}
\label{proposition variance perturbation}
There exists $\ef^* > 0$ such that, for any $0<\ef \leq \ef^*$, there is a constant $\mathscr{K} > 0$ for which
\begin{equation}
 \forall 0\leq \er \leq \frac{\ef}{2},
 \qquad
 \sigma_A^2(\er)=\sigma_A^2(0) + \mathscr{K} \er + \mathrm{O}(\er^2).
\label{eq: proposition variance perturbation}
\end{equation}
\end{proposition}

The proof can be read in Section~\ref{proof of the proposition variance perturbation}. The assumption that $\ef$ is sufficiently small ensures that Assumption~\ref{hyp THM} holds (see Section~\ref{proof of lemma arps condition}). The result is formally clear. The difficulty in proving it is that the kinetic energy is not a smooth function of~$\er$ because the shift function is only piecewise smooth.

\begin{remark}
\label{rmk:LR_var}
An inspection of the proof of Proposition~\ref{proposition variance perturbation} shows that the linear response result can be generalized to non-zero values of~$\er$ and in fact to linear responses in the parameter~$\ef$ as well. For the latter case, we consider $f_{0,\ef}(x) = f_{0,1}(x/\ef)$. Denoting now by $\sigma^2(\er,\ef)$ the variance associated with the kinetic energy~$U_{\er,\ef}$, it can be proved that, for $0 < \er < \ef$ not too large, there are $a,b \in \mathbb{R}$ such that, for $\delta,\eta \in \mathbb{R}$ sufficiently small,
\[
\sigma^2(\er+\delta,\ef+\eta) = \sigma^2(\er,\ef) + a\delta + b\eta + \mathrm{O}(\eta^2 + \delta^2).
\]
\end{remark}

%------------------------------- NUMERICS -----------------------
\section{Numerical results}
\label{sec:numerics}

The aim of this section is to quantify the evolution of the variance of AR-Langevin dynamics as the parameters of the kinetic energy function are modified. We first consider in Section~\ref{sec:Galerkin} a simple system in spatial dimension~1, for which the variance can be very precisely computed using a Galerkin-type approximation. We next consider more realistic particle systems in Section~\ref{subsection A more realistic system: second order discretization}, relying on molecular dynamics simulations to estimate the variance. In this section, the function $f_{0,\ef}(x)$ is chosen to be of the form $f_{0,1}(x/\ef)$, with $f_{0,1}$ a fifth-order spline function.

\subsection{A simple one-dimensional system}
\label{sec:Galerkin}

We first consider a single particle in spatial dimension $d=1$, in the periodic domain $\mathcal{D} = 2\pi \mathbb{T}$ and at inverse temperature $\beta = 1$. In this case, it is possible to directly approximate the asymptotic variance~\eqref{eq: variance expression} using some Galerkin discretization, as in~\cite{risken1984fokker} or~\cite{latorre2013corrections}.

We denote by $\Lope_{\er,\ef}$ the generator of the modified Langevin dynamics associated with the AR kinetic energy function $U_{\er,\ef}$ defined in~\eqref{definition U}, by $\mu_{\er,\ef}$ the associated canonical measure, and by $\Pi_{\er,\ef}$ the projector onto functions of $L^2(\mu_{\er,\ef})$ with average~0 with respect to~$\mu_{\er,\ef}$.

For a given observable~$A$, we first approximate the solution of the following Poisson equation:
\begin{equation}
\label{eq:Poisson_A}
-\Lope_{\er,\ef} \Phi_A = \Pi_{\er,\ef} A,
\end{equation}
and then compute the variance as given by~\eqref{eq: variance expression}:
\[
\sigma_A^2 = 2\int_\mathcal{E} \Phi_A \, A \, d\mu_{\er,\ef}.
\]
To achieve this, we introduce the basis functions $\psi_{nk}(q,p):= G_k(q)\,H_n(p)$, where $G_k(q) = (2\pi)^{-1/2} \mathrm{e}^{\mathrm{i} kq}$ (for $k\in\mathbb{Z}$) and $H_n(p)$ are the Hermite polynomials:
\[
\qquad H_n(p) = (-1)^n \mathrm{e}^{p^2/2}\frac{d^n}{dp^n}\Big(\mathrm{e}^{-p^2/2}\Big), \forall n\in \mathbb{N}.
\]
The choice of $G_k$ is natural in view of the spatial periodicity of the functions under consideration, while Hermite polynomials are eigenfunctions of the generator associated with the Ornstein-Uhlenbeck process on the momenta for the standard quadratic kinetic energy $p^2/2$. Note however that, when the kinetic energy is modified as $U_{\er,\ef}$, the Hermite polynomials are no longer orthogonal for the $L^2(\mu_{\er,\ef})$ scalar product.

We approximate the Poisson equation~\eqref{eq:Poisson_A} on the basis \hfill \newline$\mathcal{V}_{N_G,N_H} = \{ \psi_{nk} \}_{0 \leq n \leq N_H, \, -N_G \leq k \leq N_G}$ for given integers $N_G,H_H \geq 1$, and we look for approximate solutions of the form $\Pi_{\er,\ef}\Phi_A^{N_G,N_H}$ with
\[
\Phi_A^{N_G,N_H} = \sum_{n = -N_H}^{N_H} \sum_{k=0}^{N_G} \left[b_{N_G,N_H}\right]_{nk} \psi_{nk},
\]
where $b_{N_G,N_H} = (b_{nk})_{0 \leq n \leq N_H, \, -N_G \leq k \leq N_G}$ is a vector of size $(2N_G+1)(N_H+1)$. Restricting~\eqref{eq:Poisson_A} to $\mathcal{V}_{N_G,N_H}$ leads to
\begin{equation}
M_{N_G,N_H} b_{N_G,N_H} = a_{N_G,N_H},
\label{eq: poisson matrix galerkin}
\end{equation}
where $M_{N_G,N_H}$ is a matrix of size $(2N_G+1)(N_H+1) \times (2N_G+1)(N_H+1)$ and $a_{N_G,N_H}$ a vector of size $(2N_G+1)(N_H+1)$, whose entries respectively read
\[
\begin{aligned}
\left[M_{N_G,N_H}\right]_{nk,ml} & = \left\langle \psi_{ml},-\Lop_{\er,\ef} \psi_{nk}\right\rangle_{L^2(\mu_{\er,\ef})}, \\
\left[a_{N_G,N_H}\right]_{ml} & = \left\langle \psi_{ml},\Pi_{\er,\ef} A \right\rangle_{L^2(\mu_{\er,\ef})}. \\
\end{aligned}
\]
The approximated solution $\Phi_A^{N_G,N_H}$ of the Poisson equation \eqref{eq:Poisson_A} can therefore be computed by solving~\eqref{eq: poisson matrix galerkin}. Note however that some care is needed at this stage since $\Lope_{\er,\ef}$ is not invertible on $\mathcal{V}_{N_G,N_H}$, because the basis functions $\{ \psi_{nk} \}_{0 \leq n \leq N_H, \, -N_G \leq k \leq N_G}$ are not of integral~0 with respect to~$\mu_{\er,\ef}$. We correct this by performing a singular value decomposition of~$M_{N_G,N_H}$, removing the component of $a_{N_G,N_H}$ associated with the singular value~0, and computing the inverse of $M_{N_G,N_H}$ on the subspace generated by the eigenvectors associated with non-zero eigenvalues. In practice, we compute the entries of $a_{N_G,N_H}$ and $M_{N_G,N_H}$ by numerical quadrature. Since the Hermite polynomials are no longer orthogonal for the $L^2(\mu_{\er,\ef})$ scalar product, quadratures are required both in position and momentum variables. The variance is finally approximated as
\[
\sigma_A^2(N_G,N_H) = 2 \int_\mathcal{E} A \, \Phi_A^{N_G,N_H} \, d\mu_{\er,\ef} = 2 b_{N_G,N_H}^T a_{N_G,N_H}.
\]

In the simulations presented in this section, the potential is $V(q)=\cos(q)$, the observable under study is $A = V$, and we always set $N_H = 2N_G-1$. Figure~\ref{fig:convergence Galerkin} presents the convergence of the variance with respect to the basis size, for the standard Langevin dynamics and the AR Langevin dynamics with $\ef=2$ and various values of $\er$. The results show that the choice $N_G=12$ is sufficient in all cases to approximate the asymptotic value. We checked in addition in one case, namely for the standard dynamics, that the values we obtain are very close to a reference value obtained with $N_G=30$: the relative variation is of order $10^{-8}$ for $N_G=10$, $10^{-10}$ for $N_G=12$  and $10^{-11}$ for $N_G=14$. We therefore set $N_G=12$ in the remainder of this section.

\begin{figure}[t]
 % \vspace{-100pt}
  \centering
  \includegraphics[width=0.80\textwidth]{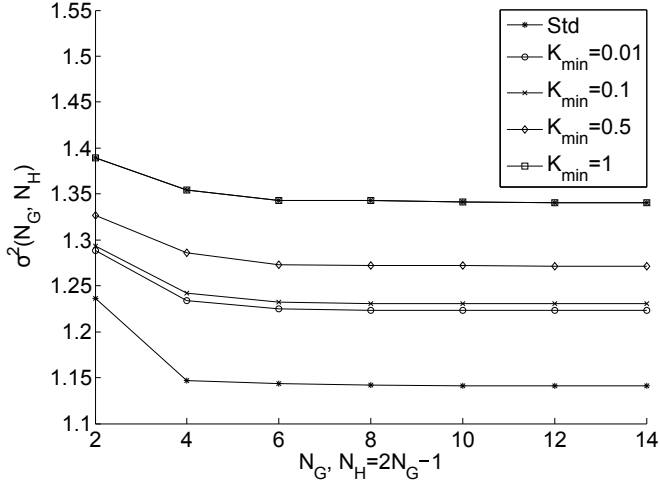}
%  \vspace{90pt}
  \caption{Convergence of the Galerkin approximation  in the basis size $N_G$ and $N_H=2N_G-1$: approximation of the variance of observable $A=V$ for the standard dynamics and the AR dynamics with fixed parameter $\ef=2$ and various values of $\er$.}
  \label{fig:convergence Galerkin}
	\end{figure}
		
The variation of the computed variance for $A=V$ is plotted in Figure~\ref{fig:varianceGalerkinAllKmaxKmin} for various parameters $0\leq \er<\ef$ of the AR-Langevin dynamics. Note that, as expected, the variance increases with increasing values of $\er$ for fixed $\ef$, but also with increasing values of $\ef$ for fixed $\er$. We next illustrate the linear response results of Proposition~\ref{proposition variance perturbation} and Remark~\ref{rmk:LR_var} in Figures~\ref{fig:Variance_Kmax5_overKmin left} and ~\ref{fig:Variance_Kmax5_overKmin right}: in both situations, the variance increases linearly with the parameter under consideration is varied in a sufficiently small neighborhood of its initial value. After that initial regime, nonlinear variations appear. Note also that the relative increase of the variance is more pronounced as a function of $\ef$ than $\er$.

\begin{figure}[t]
 % \vspace{-100pt}
  \centering
  \includegraphics[width=0.80\textwidth]{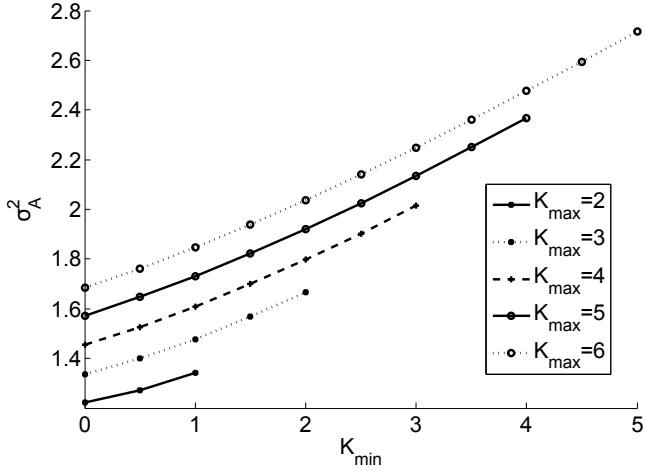}
  %\vspace{90pt}
  \caption{Asymptotic variance of time averages for $A=V$, approximated by the Galerkin method, as a function of $\er$ and for several values of~$\ef$.}
 \label{fig:varianceGalerkinAllKmaxKmin}
\end{figure}

\begin{figure}[t]
 % \vspace{-100pt}
  \centering
  \includegraphics[width=0.8\textwidth]{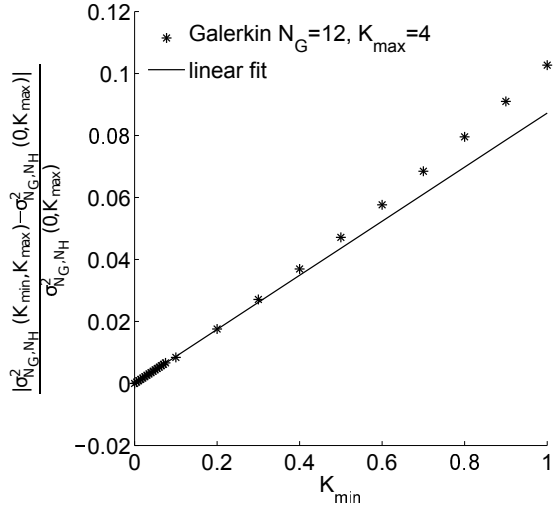}
 % \vspace{80pt}
  \caption{Relative difference between the variance $\sigma^2_{N_G,N_H}(\er,\ef)$ and its initial value computed for reference parameters. (Left) Fixed upper bound $\ef=4$, and reference value $\er = 0$.}
  \label{fig:Variance_Kmax5_overKmin left}
	\end{figure}

\begin{figure}[t]
 % \vspace{-90pt}
  \centering
  \includegraphics[width=0.8\textwidth]{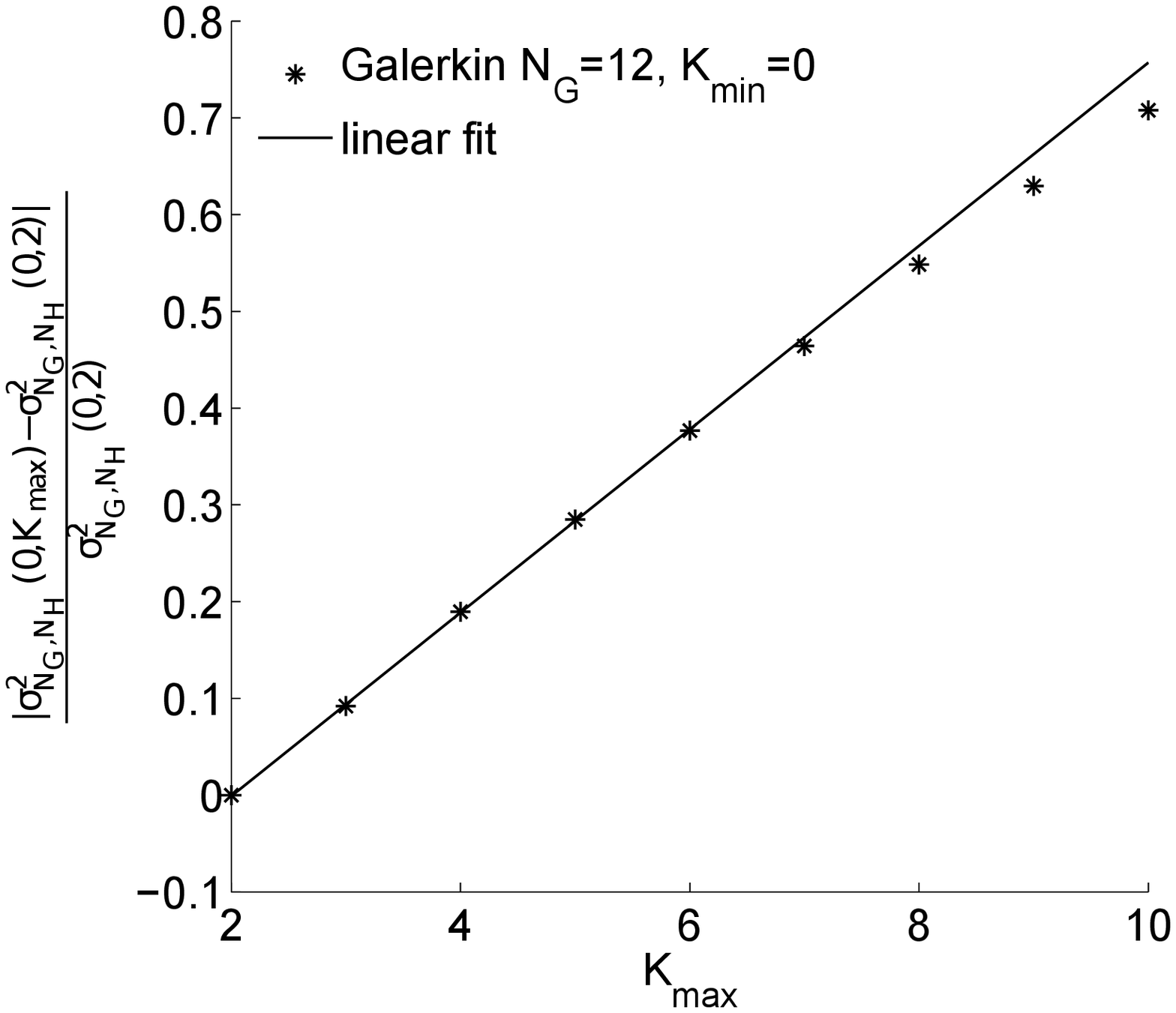}
 % \vspace{80pt}
  \caption{Same as Figure \ref{fig:Variance_Kmax5_overKmin left}. Fixed lower bound $\er=0$, and reference value $\ef = 2$.}
  \label{fig:Variance_Kmax5_overKmin right}
	\end{figure}

%%%% OLD FORMAT FIGURES
%%%\begin{figure}
  %%%\centering
  %%%\vspace{-50pt}
  %%%\begin{subfigure}[b]{0.5\textwidth}
    %%%\includegraphics[width=\textwidth]{Figures/galerkin_Kmin_Kmax4.eps}%varianceGalerkinKmin.eps}
    %%%%%\vspace{-135pt}
  %%%\end{subfigure}%
  %%%~
  %%%%\hspace{-20pt}
  %%%\begin{subfigure}[b]{0.5\textwidth}
    %%%%%\vspace{-220pt}
    %%%\includegraphics[width=\textwidth]{Figures/galerkin_Kmin0_Kmax.eps}
    %%%%			%\vspace{-100pt}
  %%%\end{subfigure}
  %%%%%\vspace{10pt}
  %%%\vspace{50pt}
  %%%\caption{Relative difference between the variance $\sigma^2_{N_G,N_H}(\er,\ef)$ and its initial value computed for reference parameters. (Left) Fixed upper bound $\ef=4$, and reference value $\er = 0$. (Right) Fixed lower bound $\er=0$, and reference value $\ef = 2$.}
  %%%\label{fig:Variance_Kmax5_overKmin}
%%%\end{figure}

\begin{remark}
\label{remark smoothness stability}
In practice, the idea usually is to set the lower bound $\er$ sufficiently large when performing Monte Carlo simulations, in order to decrease as much as possible the computational cost. The gap $\ef-\er$ should however not be too small in order to have a sufficiently smooth transition from a vanishing kinetic energy to a quadratic one. This requires therefore $\ef$ to be quite large if $\er$ is large. The results presented in Figure~\ref{fig:Variance_Kmax5_overKmin right} suggest that this may not be the optimal choice, unless the algorithmic speed-up is quite large.
\end{remark}

\subsection{A more realistic system}
\label{subsection A more realistic system: second order discretization}

In order to study the variation of the variance as a function of $\er$ and $\ef$ in systems of higher dimensions, we resort to Monte Carlo simulations. This requires discretizing the AR-Langevin dynamics~\eqref{modified Langevin}, and we resort to a scheme of weak order~2, obtained by a splitting strategy where the generator of the modified Langevin dynamics~\eqref{eq: generator modified Langevin} is decomposed into three parts:
\[
A:=\n U(p) \cdot \n_q, \qquad B:=-\n V(q)\cdot \n_p, \qquad C:= -\n U(p)\cdot \n_p + \frac{1}{\beta }\de_p \,.
\]
The transition kernel obtained by a Strang splitting reads $P_{\de t}=\mr{e}^{\gamma \de t C/2}\mr{e}^{\de t B/2 } \mr{e}^{\de t A }\mr{e}^{\de t B/2 } \mr{e}^{\gamma \de t C/2}$. Contrarily to the standard kinetic energy functions, the elementary evolution associated with $C$ cannot be integrated analytically. To preserve the order of the scheme, we approximate $\mr{e}^{\gamma \de t/2 C }$ by a midpoint rule, encoded by a transition kernel $P_{\de t}^{\gamma,C}$ satisfying $P_{\de t}^{\gamma,C}\varphi=\mr{e}^{\gamma \de t C}\varphi +\mr{O}(\de t^3) $ for smooth test functions $\varphi$. This gives the following discretization scheme:
\[
\displaystyle
\left\{
\begin{aligned}
\ds p^{n+1/4}& = p^n - \gamma\n U \left(\frac{p^{n+1/4}+p^n}{2}\right)\frac{\de t}{2} + \sqrt{\frac{\gamma\de t}{\beta}}G^n,\\
p^{n+1/2} &=  p^{n+1/4}-\n V(q^n)\frac{\de t}{2}, \\
q^{n+1}&=q^n+\n U(p^{n+1/2})\de t, \\
p^{n+3/4} &=  p^{n+1/2}-\n V(q^{n+1})\frac{\de t}{2}, \\
\ds p^{n+1}& = p^{n+3/4} - \gamma\n U \left(\frac{p^{n+1}+p^{n+3/4}}{2}\right)\frac{\de t}{2} + \sqrt{\frac{\gamma\de t}{\beta}}G^{n+1/2},
\end{aligned}
\right.
\]
where $G^n,G^{n+1/2}$ are i.i.d. standard $d$-dimensional Gaussian random variables. The first and the last line are obtained by implicit schemes, solved in practice by a fixed point strategy (the termination criterion being that the distance between successive iterates is smaller than $10^{-10}$, and the initial iterate being obtained by a Euler-Maruyama step). By following the same approach as in~\cite{Matthews}, it can indeed be proved that this scheme is of weak order~2; see~\cite{NumericsInPreparation} for further precisions.

The ergodicity of some second-order schemes was proved for the standard Langevin dynamics in \cite{Matthews}. Since the AR-Langevin dynamics can be seen as a perturbation of the standard Langevin dynamics, it can be proved by combining the proofs from~\cite{Matthews} and the proof of Theorem~\ref{thm uniform convergence} that, when $0\leq \er<\ef$ are sufficiently small, the corresponding discretization of the AR-Langevin dynamics remains ergodic (see \cite{NumericsInPreparation}). The corresponding invariant measure is denoted by $\mu_{\er, \de t}$. It also follows by the results of~\cite{Matthews} that the error on averages of smooth observables $\varphi \in \mathscr{S}$ with respect to $\mu_{\er, \de t}$ is of order~2, i.e. there exists $a \in \R$ such that
\[
\int_{\El}\varphi \, d\mu_{\er,\de t}=\int_{\El}\varphi \, d\mu_{\er} + a\de t^2 + \mr{O}\left(\de t^3\right).
\]
As already mentioned in Remark~\ref{remark smoothness stability}, the reduction of the gap between the parameters $\er$ and $\ef$ reduces the smoothness of the transition between the restrained dynamics and the full dynamics. This raises issues in the stability of the scheme, which can be partly cured by resorting to a Metropolis-Hastings correction (\cite{MRRTT53,Hastings70} and ~\cite{NumericsInPreparation}).

The system we consider is composed of $N=49$ particles in dimension~2, so that $d=98$ and $\mathcal{D} = (L\mathbb{T})^{2N}$. The masses are set to~1 for all particles. Among these particles, two particles (numbered 1 and 2 in the following) are designated to form a dimer while the others are solvent particles. All particles, except the two particles forming the dimer, interact through the purely repulsive WCA pair potential, which is a truncated Lennard-Jones potential~\cite{SBB88}:
\[
V_{\rm WCA}(r) = \left \{ \begin{array}{cl}
\displaystyle 4 \varepsilon \left [ \left ( \frac{\sigma}{r} \right )^{12}
  - \left ( \frac{\sigma}{r}\right )^6 \right ] + \varepsilon & \quad {\rm if \ } r \leq r_0, \\
0 & \quad {\rm if \ } r > r_0,
\end{array} \right.
\]
where $r$ denotes the distance between two particles, $\varepsilon$ and $\sigma$ are two positive parameters and $r_0=2^{1/6}\sigma$. The interaction potential between the two particles of the dimer is a double-well potential
\begin{equation}
  \label{intro:eq:VS}
  V_{\rm D}(r) = h \left [ 1 - \frac{(r-r_0-w)^2}{w^2} \right ]^2,
\end{equation}
where $h$ and $w$ are two positive parameters. The potential $V_{\rm D}$ has two energy minima. The first one, at $r=r_0$, corresponds to the compact state. The second one, at $r=r_0+2w$, corresponds to the stretched state. The total energy of the system is therefore, for $q \in (L\mathbb{T})^{dN}$ with $d=2$,
\[
V(q) = V_{\rm D}(|q_1-q_2|) + V_{\rm SS}(q_3,\dots,q_N) + V_{\rm DS}(q),
\]
where the solvent-solvent and dimer-solvent potential energies respectively read
\[
V_{\rm SS}(q_3,\dots,q_N) = \sum_{3 \leq i<j \leq N} V_{\rm WCA}(|q_i-q_j|),
\qquad
V_{\rm DS}(q) = \sum_{i=1,2} \sum_{3 \leq j \leq N} V_{\rm WCA}(|q_i-q_j|).
\]
We choose $\beta = 1$, $\e_{LJ}=1$, $\sigma_{LJ}=1$, $h=1$, $w=1$, and set the particle density $\rho=N/L^2$ to~0.56 in the numerical results presented in this section, sufficiently high to ensure that the solvent markedly modifies the distribution of configurations of the dimer compared to the gas phase.

The source of metastability in the system is the double-well potential on the dimer. In such a system, it makes sense to restrain only solvent particles (since they account for most of the computational cost), and keep the standard kinetic energy for the particles forming the dimer (since the observable depends on their positions). As noted in Remark~\ref{remark individual parameters}, the method allows us to choose different individual kinetic energies for different particles. Since the solvent interacts with the dimer, we study how the variance of time averages of observables related to the configuration of the dimer, such as the dimer potential energy~$A=V_{\rm D}$, depend on the restraining parameters chosen for the solvent particles. We also estimate the variance of time averages based on observables depending only on the solvent degrees of freedom, such as the solvent-solvent potential energy $A=V_{\rm SS}$.

The asymptotic variance of time averages for a given observable~$A$ is estimated by approximating the integrated auto-correlation function
\[
\sigma^2_A=2\int_0^{\infty}\E_{\mu_{\er,\ef}}\left[\left(\Pi_{\mu}A\right)(q_0,p_0)\left(\Pi_{\mu}A\right)(q_t,p_t)\right]dt,
\]
where the expectation is with respect to initial conditions $(q_0,p_0)\sim \mu_{\er,\ef}$ and all realizations of the AR Langevin dynamics. This is done by first truncating the upper bound in the integral by a sufficiently large time $T_{\mr{corr}}$, and using a trapezoidal rule:
\[
\sigma^2_A\approx \sigma^2_{A,M,\de t} := \de t\left(\frac{\widetilde{C}_0^M}{2}+\sum_{j=1}^{I_{\rm{corr}}} \widetilde{C}_j^M\right)
\]
where $I_{\rm{corr}}=\left\lfloor \frac{T_{\mr{corr}}}{\de t}\right\rfloor$, and the empirical averages over $M$ realizations of trajectories of $I_{\rm{corr}}$ steps are defined as
\[
\qquad\widetilde{C}^M_j := C^M_j - \widehat{A}_j^M \widehat{A}_0^M , \quad j\in \left\{1, \ldots, I_{\mr{corr}} \right\}\,,
\]
with
\[
C_j^M :=
\frac{1}{M}\sum_{m=1}^M A(q_j^m,p_j^m)A(q_0^m,p_0^m), \quad \widehat{A}_j^M := \frac{1}{M}\sum_{m=1}^M A(q_j^m,p_j^m)\,.
\]
The initial condition $(q_0^{m+1},p_0^{m+1})$ for the $m+1$th trajectory is obtained from the last configuration of the $m$th configuration, namely $(q_{I_{\rm{corr}}}^{m+1},p_{I_{\rm{corr}}}^{m+1})$. Figure~\ref{fig:acf} presents the auto-correlation function obtained for $A=V_{\rm D}$. The results show that the choice $T_{\rm corr} = 3$ is reasonable.

%%%%%% FIGURES FORMAT
%%%%\begin{figure}[t]
  %%%%%\vspace{-100pt}
  %%%%\centering
  %%%%\includegraphics[width=0.50\textwidth]{Figures/acf.eps}
 %%%%\vspace{70pt}
  %%%%\caption{Auto-correlation function $\E_{\mu}\left[\left(\Pi_{\mu}A\right)(p_0,q_0)\left(\Pi_{\mu}A\right)(p_t,q_t)\right]$ for $A=V_{D}$ as a function of time.}
 %%%%\label{fig:acf}
%%%%\end{figure}

\begin{figure}[t]
 % \vspace{-90pt}
  \centering
  \includegraphics[width=0.8\textwidth]{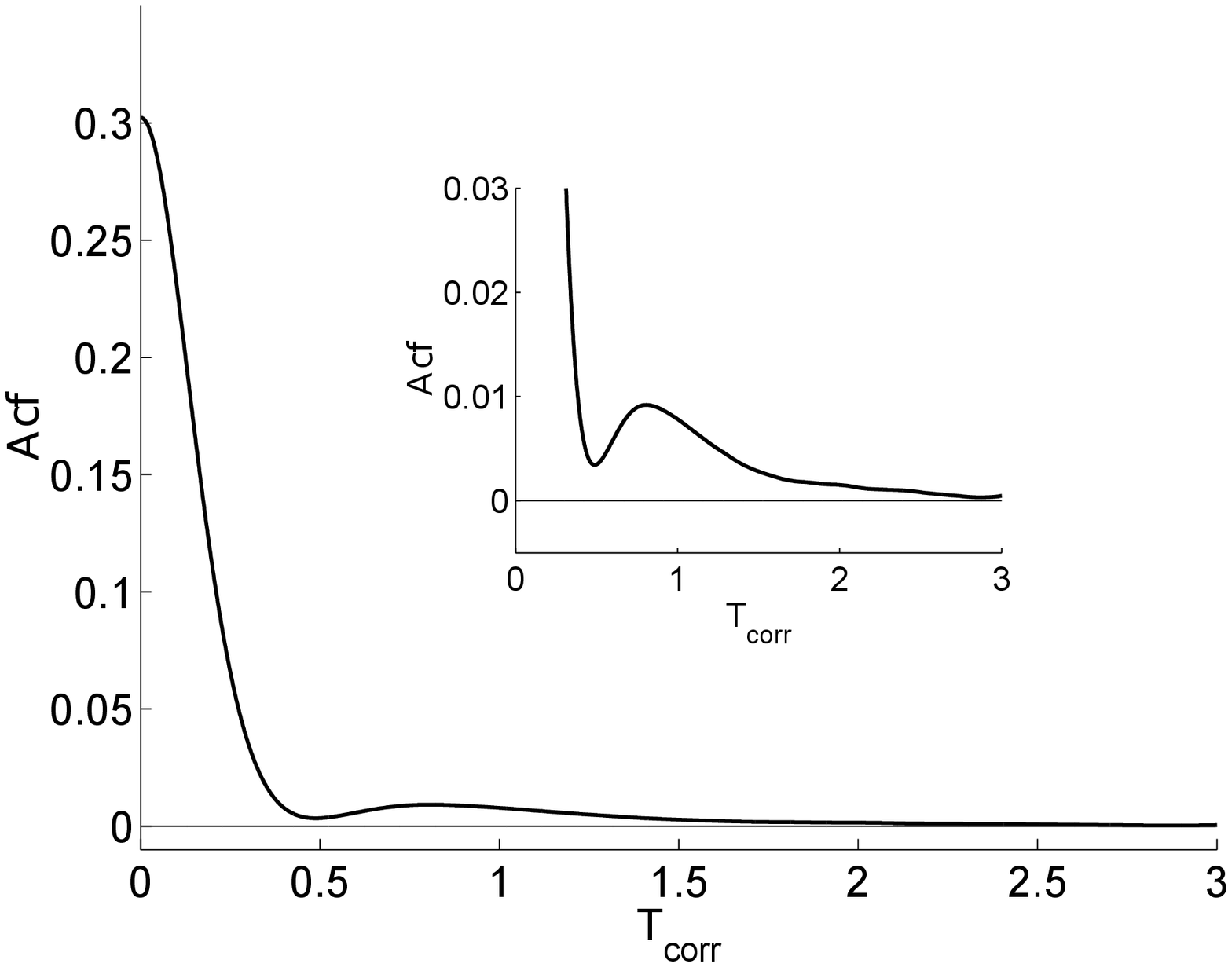}
  %\vspace{80pt}
  \caption{Auto-correlation function $\E_{\mu}\left[\left(\Pi_{\mu}A\right)(p_0,q_0)\left(\Pi_{\mu}A\right)(p_t,q_t)\right]$ for $A=V_{D}$ as a function of time.}
  \label{fig:acf}
	\end{figure}

The results of~\cite{Matthews,NumericsInPreparation} show that the errors on the approximation of the variance should be of order~$\de t^2$ when $T_{\rm corr} \to +\infty$. This is illustrated in Figures~\ref{fig:variance zero coeff left},~\ref{fig:variance zero coeff right},~\ref{fig:variance max coeff left} and~\ref{fig:variance max coeff right}, which present the convergence of $\sigma^2_{A,M,\de t}$ as a function of $\de t$ for $M = 3 \times 10^6$.
%% THE NUMBER OF TIME STEPS IS Nsteps= 10^10-> THEN DEPENDING ON THE TIME STEP SIZE DT THE NUMBER OF M IS M=Nsteps*DT/Tcorr!
 It is possible to extrapolate the value of the variance at $\de t=0$ by fitting $\sigma^2_{A,M,\de t}$ as $a_0+a_1\de t^2$. Note that the errors on the variance are bigger in the case $\er=2.7$, which is expected due to the smaller gap between the parameters $\ef,\er$. In the sequel, all the reported approximations of the variance are obtained by computing $\sigma^2_{A,M,\de t}$ for~6 values of the time step $\de t$, and extrapolating to the limit $\de t \to 0$ as in Figures~\ref{fig:variance zero coeff left},~\ref{fig:variance zero coeff right},~\ref{fig:variance max coeff left} and~\ref{fig:variance max coeff right}. More precisely, the time steps are chosen as $\de t_{0,k} =k \times 10^{-3}$ for $k=1,\dots,6$ when $\er = 0$, and $\de t_{\er^*,k}=k \times 10^{-4}$ for $\er^* = 2.7$. For intermediate values of $\er$, the time steps $\de t_{\er,k}$ are obtained by a linear interpolation between $\de t_{0,k}$ and $\de t_{\er^*,k}$.

\begin{figure}[t]
%  \vspace{-90pt}
  \centering
  \includegraphics[width=0.8\textwidth]{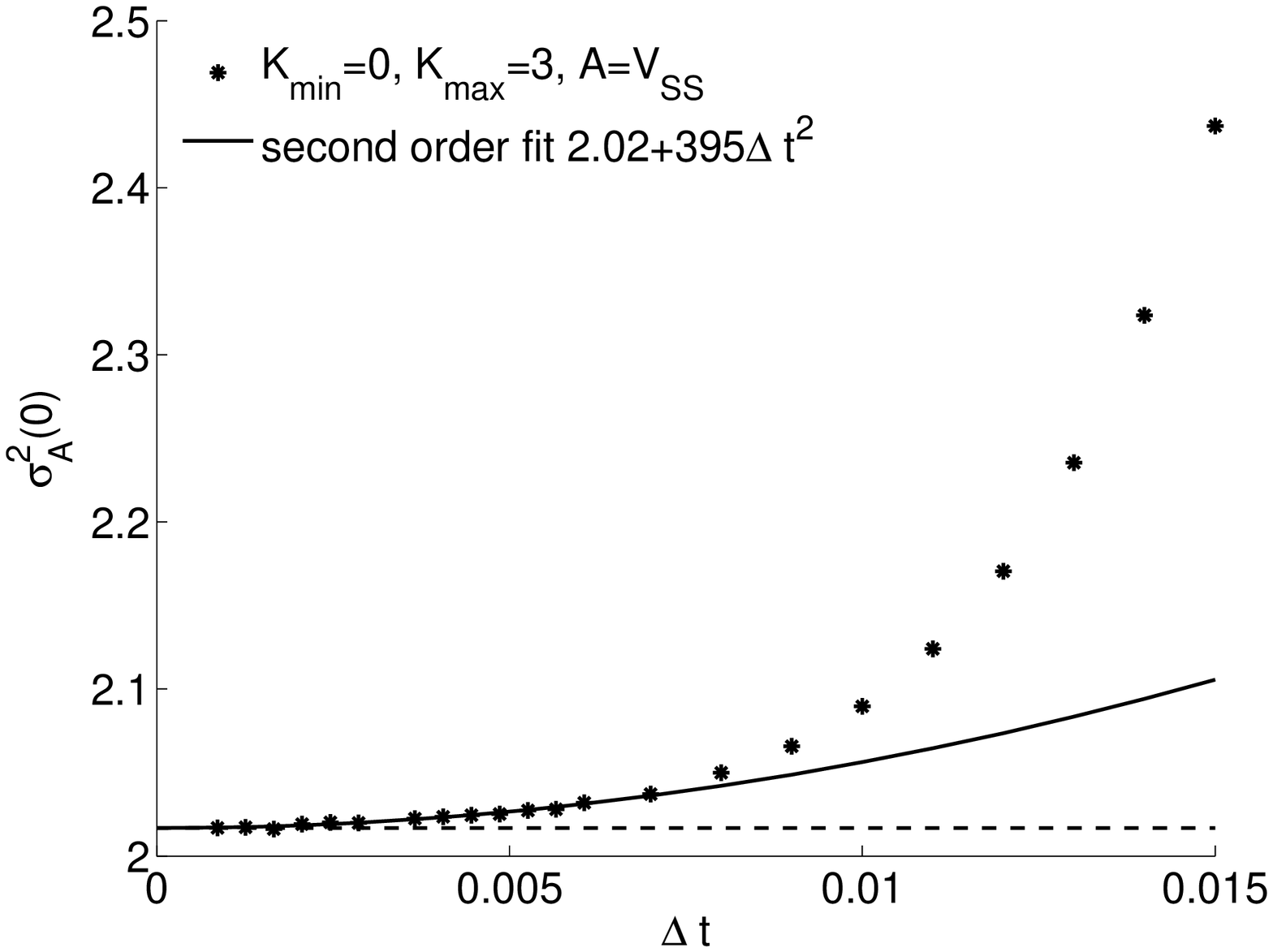}
 % \vspace{80pt}
  \caption{Estimated variance $\sigma^2_{A,M,\de t}$ for $A=V_{\mr{SS}}$, as a function of the step size $\de t$, for $\er=0$ and $ \ef=3$.}
  \label{fig:variance zero coeff left}
	\end{figure}

\begin{figure}[t]
 % \vspace{-90pt}
  \centering
  \includegraphics[width=0.8\textwidth]{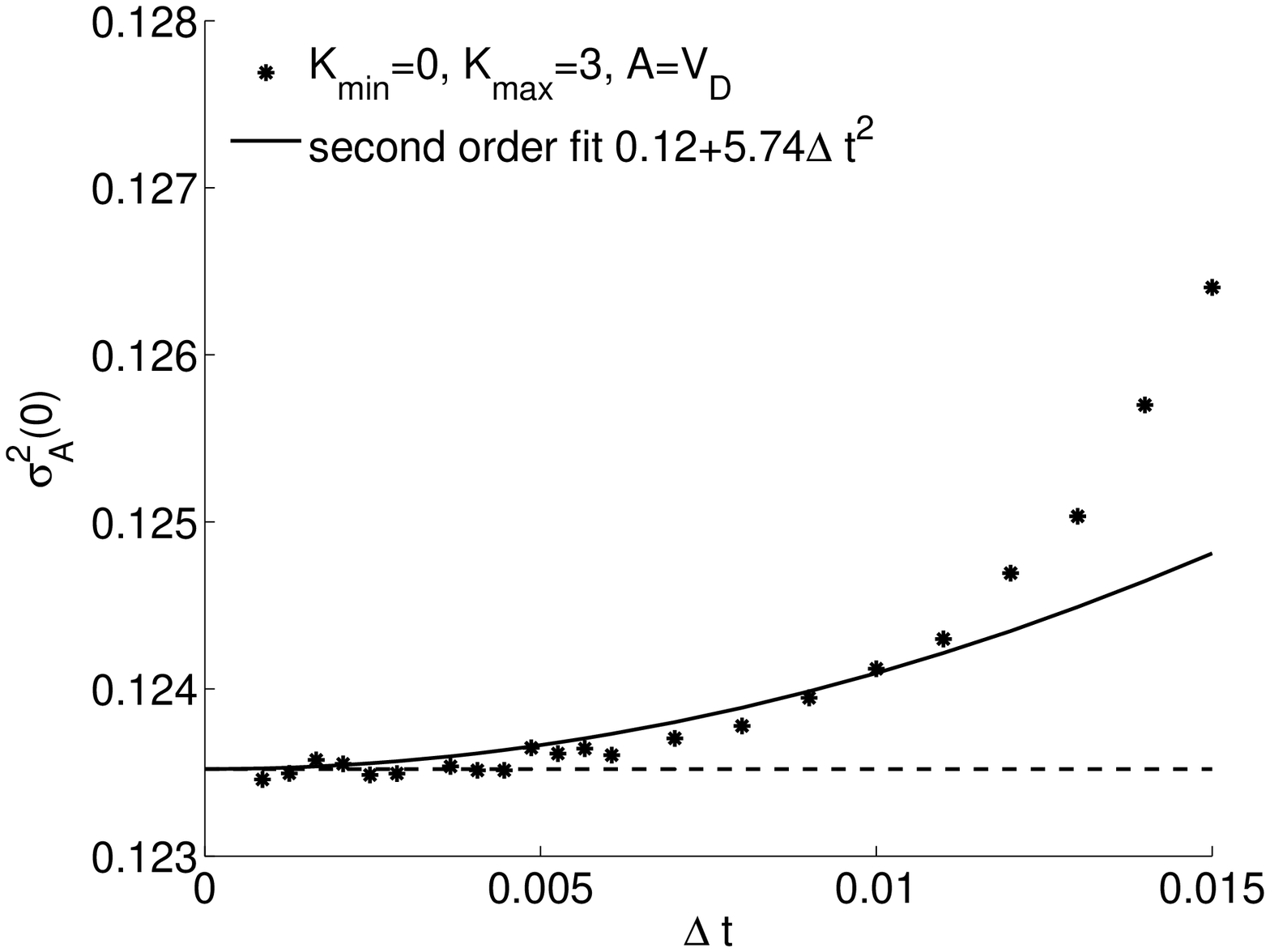}
%  \vspace{80pt}
  \caption{Estimated variance $\sigma^2_{A,M,\de t}$ for $A = V_{\mr{D}}$, as a function of the step size $\de t$, for $\er=0$ and $ \ef=3$.}
  \label{fig:variance zero coeff right}
	\end{figure}

%%%% ORIGINAL FORMAT OF FIGURES
%%%\begin{figure*}%[t]
  %%%\centering
  %%%%\vspace{-50pt}
  %%%\begin{subfigure}[b]{0.5\textwidth}
    %%%\includegraphics[width=\textwidth]{Figures/zero_coeff_variance_in_dt_SS.eps}
  %%%\end{subfigure}%
  %%%~       				%\hspace{-20pt}
  %%%\begin{subfigure}[b]{0.5\textwidth}
    %%%%		%\vspace{-220pt}
    %%%\includegraphics[width=\textwidth]{Figures/zero_coeff_variance_in_dt_DD.eps}
  %%%\end{subfigure}
  %%%\vspace{50pt}
  %%%\caption{Estimated variance $\sigma^2_{A,M,\de t}$ for $A=V_{\mr{SS}}$ (Left) and $A = V_{\mr{D}}$ (Right), as a function of the step size $\de t$, for $\er=0$ and $ \ef=3$.}
  %%%\label{fig:variance zero coeff}
%%%\end{figure*}

\begin{figure}[t]
  %\vspace{-90pt}
  \centering
  \includegraphics[width=0.8\textwidth]{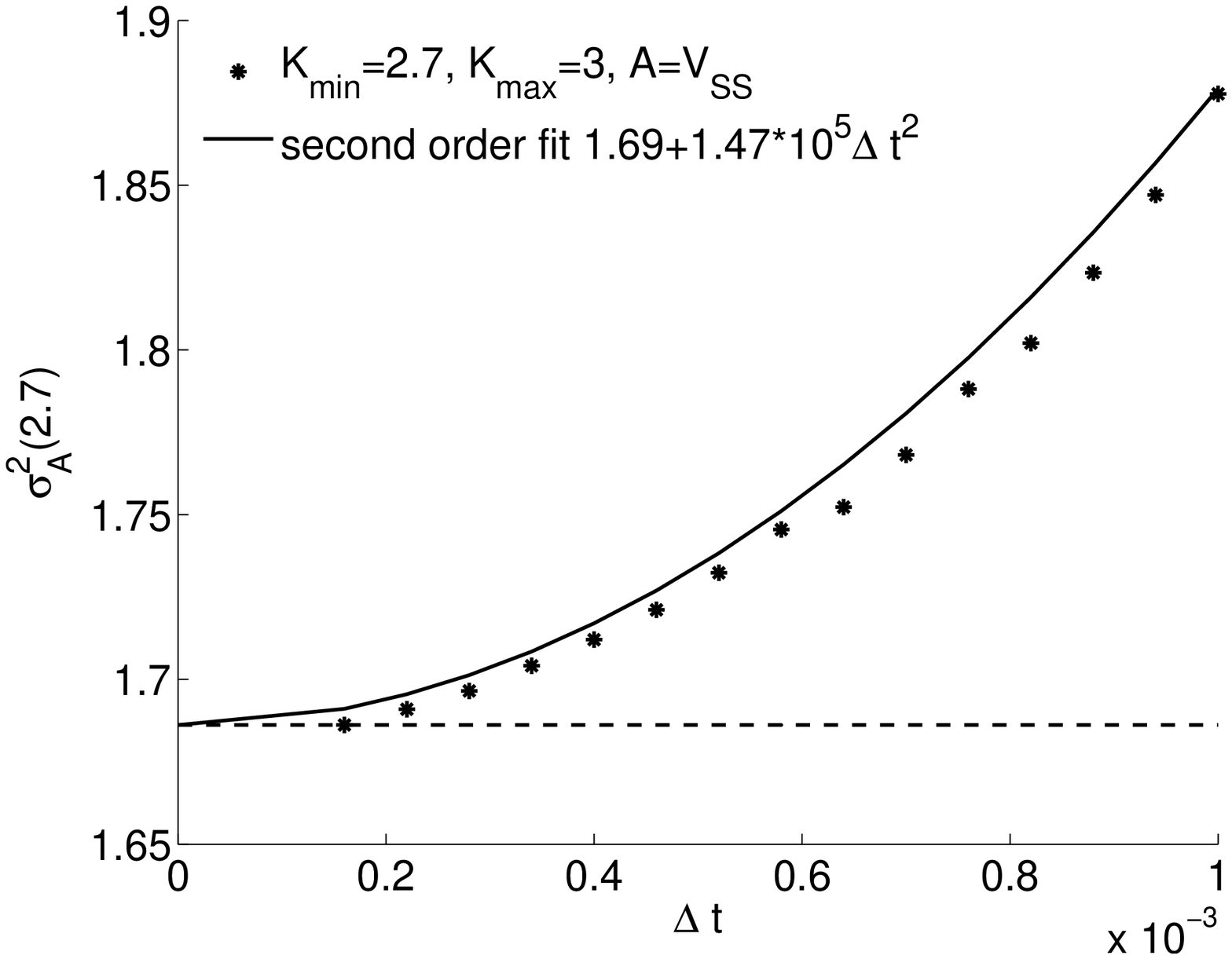}
 % \vspace{90pt}
  \caption{Same as Figure~\ref{fig:variance zero coeff left}, except that $\er=2.7$.}
  \label{fig:variance max coeff left}
	\end{figure}

\begin{figure}[t]
%  \vspace{-90pt}
  \centering
  \includegraphics[width=0.8\textwidth]{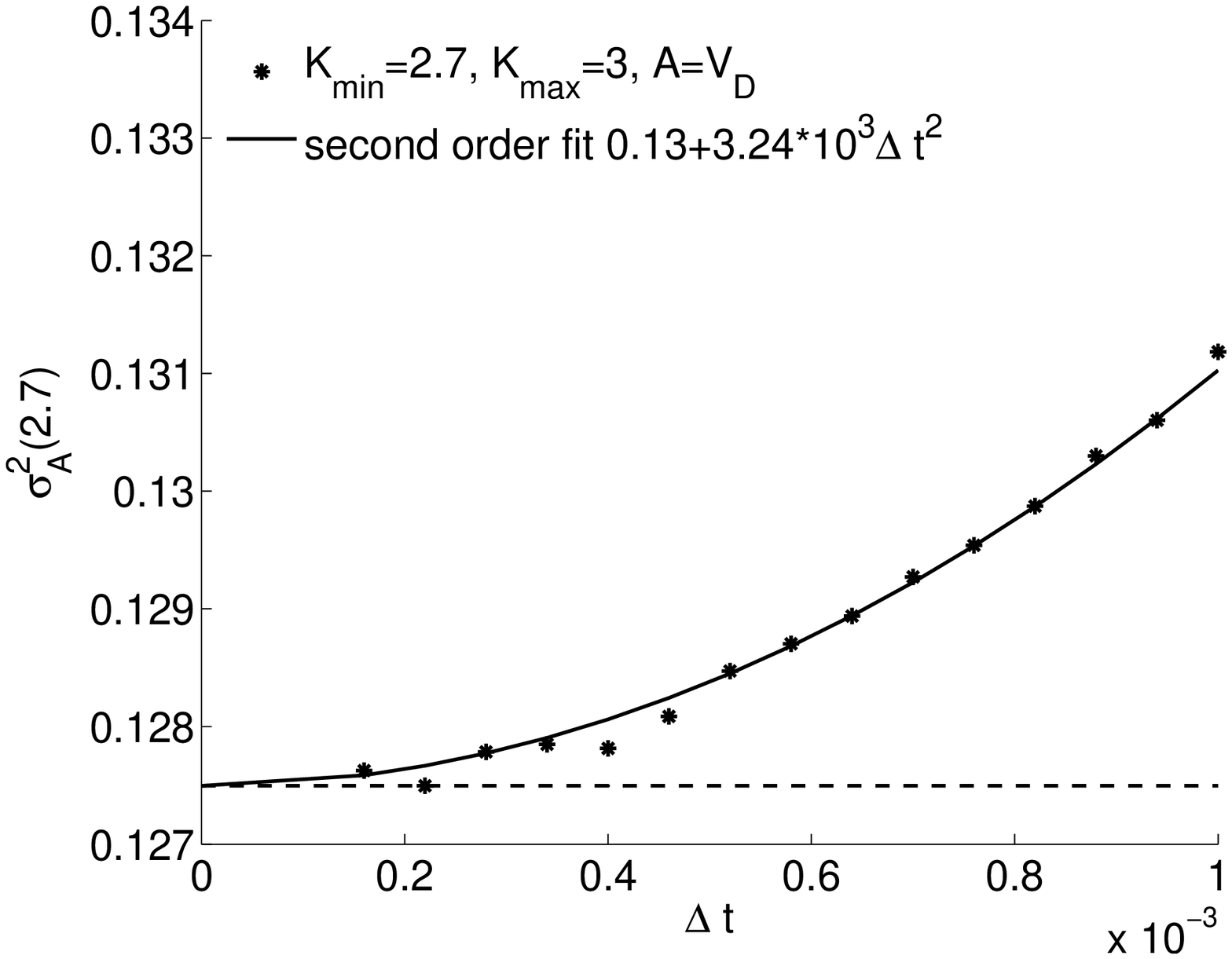}
%  \vspace{90pt}
  \caption{Same as Figure~\ref{fig:variance zero coeff right}, except that $\er=2.7$.}
  \label{fig:variance max coeff right}
	\end{figure}

%%%%% original format figures
%%%%\begin{figure}[t]
  %%%%\centering
  %%%%%%\vspace{-50pt}
  %%%%\begin{subfigure}[b]{0.5\textwidth}
    %%%%\includegraphics[width=\textwidth]{Figures/max_coeff_variance_in_dt_SS.eps}
    %%%%%	%\vspace{-135pt}
    %%%%%	\label{fig:ACF}
    %%%%%	\caption{}
  %%%%\end{subfigure}%
  %%%%~       				%\hspace{-20pt}
  %%%%\begin{subfigure}[b]{0.5\textwidth}
    %%%%%		%\vspace{-220pt}
    %%%%\includegraphics[width=\textwidth]{Figures/max_coeff_variance_in_dt_DD.eps}
    %%%%%	%\vspace{-100pt}
    %%%%%\label{fig:VarianceDimerPOtentialOverEps}
    %%%%%	\caption{Variance of the dimer potential energy for $\ef = ?$ over $\er$.}
  %%%%\end{subfigure}
 %%%%\vspace{50pt}
  %%%%\caption{Same as Figure~\ref{fig:variance zero coeff}, except that $\er=2.7$.}
  %%%%\label{fig:variance max coeff}
%%%%\end{figure}

The variations as a function of~$\er$ of the approximations of the variances $\sigma^2_{A}(\er)$ for the solvent-solvent potential energy $V_{\rm SS}$ and the dimer potential energy $V_{\rm D}$ are reported in Figures~\ref{fig:variance over kmin left} and \ref{fig:variance over kmin right}. Surprisingly, even though the solvent particles are restrained, the variance of the solvent-solvent potential decreases linearly for moderately small values of $\er$; whereas, as expected, the variance of the dimer potential, which is only implicitly influenced by the restraining parameters, increases linearly for these values of $\er$. In order to more easily compare the impacts of the restraining procedure, we plot in Figure~\ref{fig:comparisonRelDifferenceVariance_SS_DD left} the relative differences of the variance $\sigma^2(\er)$ and the variance of Zero-$\ef$-AR dynamics $\sigma_A^2(0)$ as a function of $\er$. For the two observables under consideration, the impact of an increase of the parameter $\er$ on the variance associated with the dimer potential is much weaker than on the variance related to the solvent potential. We also provide in Figure~\ref{fig:comparisonRelDifferenceVariance_SS_DD right} the percentage of restrained particles, which directly depends on the restraining parameter $\er$ and dictates the algorithmic speed-up. This supports the idea that the use of the AR-Langevin method for heterogeneous systems can be beneficial when the AR parameters are set to non-zero values for the part of the system which is not directly of interest (e.g. the solvent), while the standard kinetic energy should be kept for the degrees of freedom that are directly involved in the observable (e.g. the dimer).

\begin{figure}[t]
 % \vspace{-90pt}
  \centering
  \includegraphics[width=0.8\textwidth]{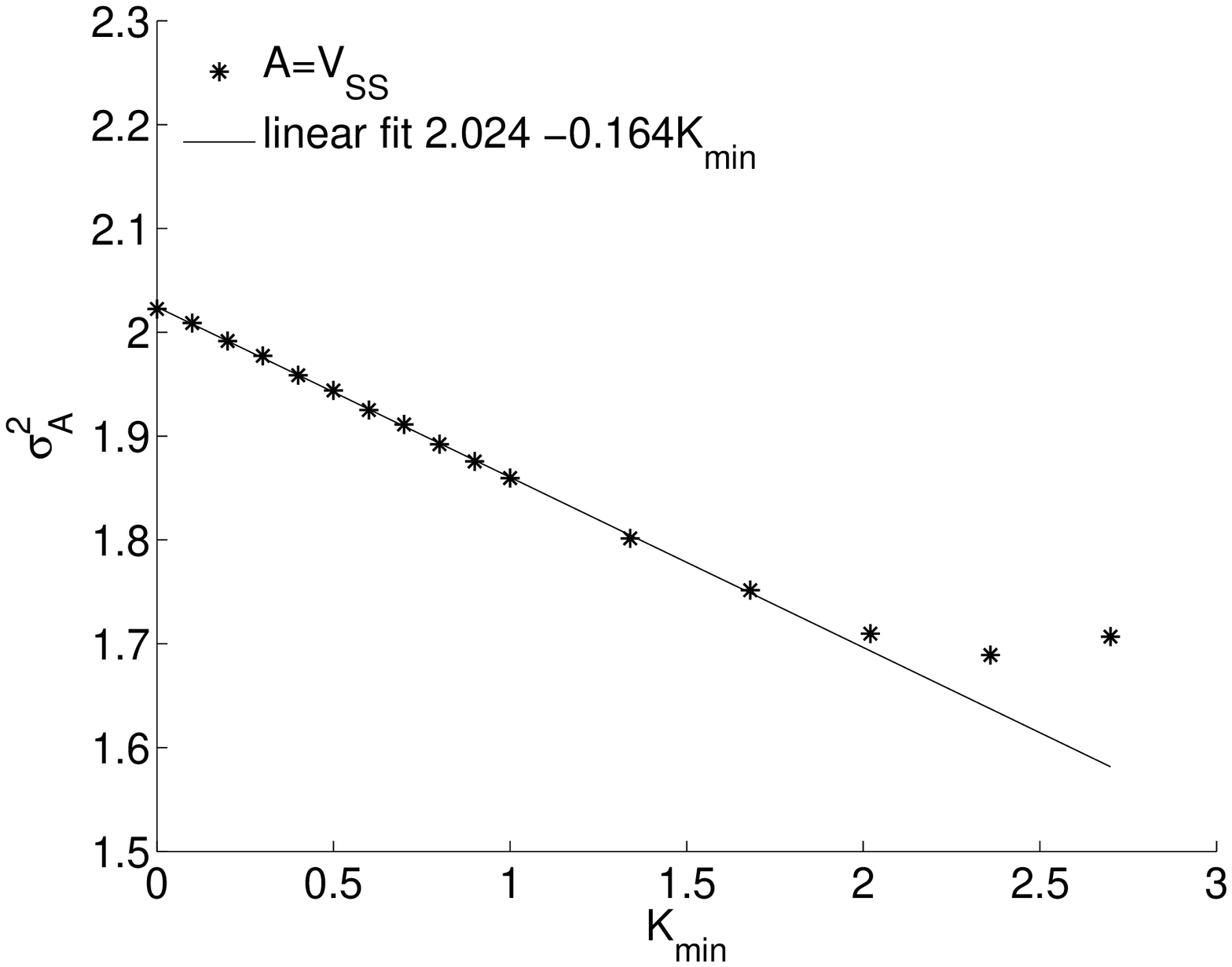}
%  \vspace{90pt}
  \caption{Estimated variance $\sigma^2_{A}$ for $A=V_{\mr{SS}}$ as a function of $\er\in \left[0,2.7\right]$ for $\ef=3$.}
  \label{fig:variance over kmin left}
	\end{figure}

\begin{figure}[t]
 % \vspace{-90pt}
  \centering
  \includegraphics[width=0.8\textwidth]{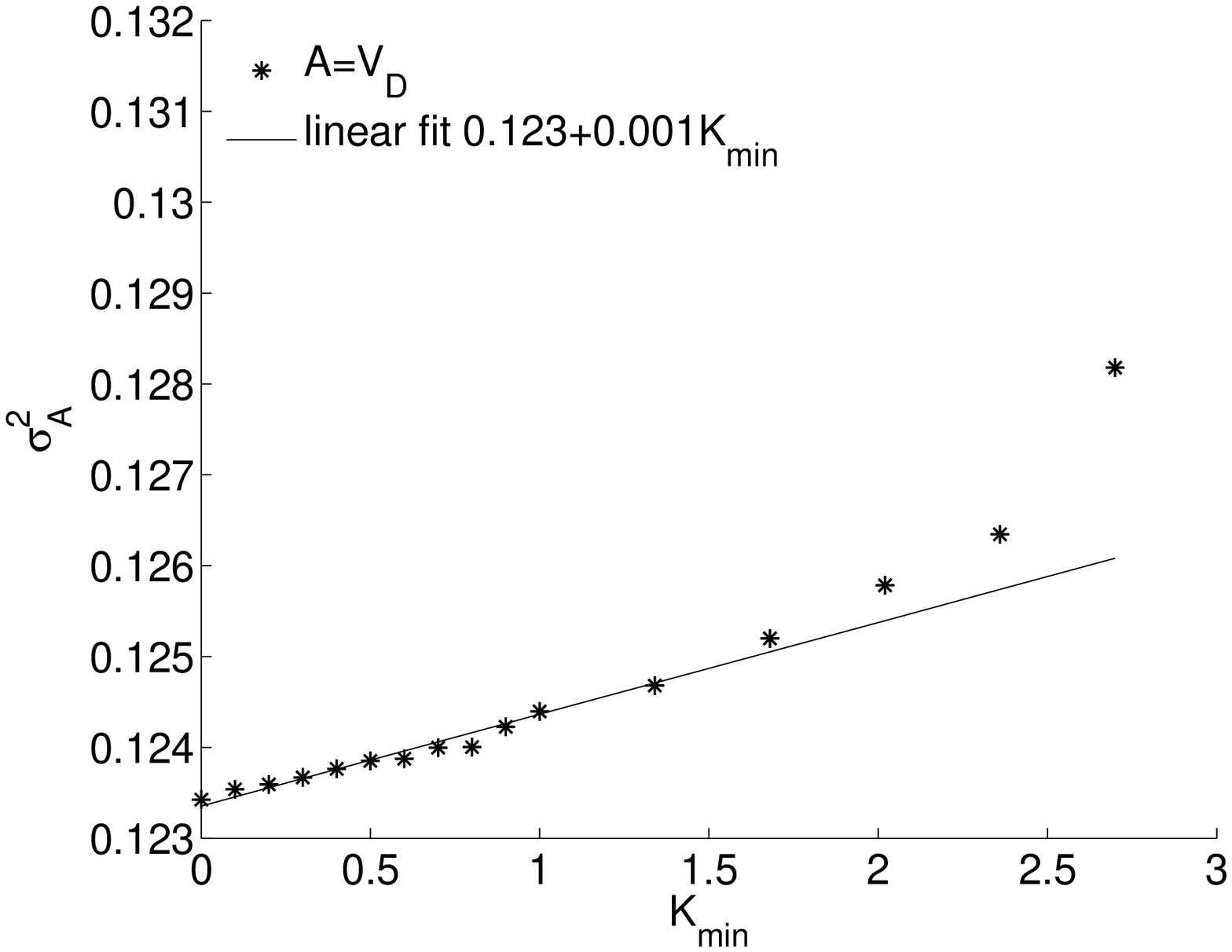}
 % \vspace{90pt}
  \caption{Estimated variance $\sigma^2_{A}$ for $A=V_{\rm D}$ as a function of $\er\in \left[0,2.7\right]$ for $\ef=3$.}
  \label{fig:variance over kmin right}
	\end{figure}

%%%%%%% original format figure
%%%%%\begin{figure}[t]
  %%%%%\centering
  %%%%%%\vspace{-120pt}
  %%%%%\begin{subfigure}[b]{0.5\textwidth}
    %%%%%\includegraphics[width=\textwidth]{Figures/variance_over_kmin_SS.eps}
    %%%%%%	%\vspace{-135pt}
    %%%%%%	\label{fig:ACF}
    %%%%%%	\caption{}
  %%%%%\end{subfigure}%
  %%%%%~       				%\hspace{-20pt}
  %%%%%\begin{subfigure}[b]{0.5\textwidth}
    %%%%%%		%\vspace{-220pt}
    %%%%%\includegraphics[width=\textwidth]{Figures/variance_over_kmin_DD.eps}
  %%%%%\end{subfigure}
%%%%%\vspace{50pt}
  %%%%%\caption{Estimated variance $\sigma^2_{A}$ for $A=V_{\mr{SS}}$ (Left) and $A=V_{\rm D}$ (Right) as a function of $\er\in \left[0,2.7\right]$ for $\ef=3$.}
  %%%%%\label{fig:variance over kmin}
%%%%%\end{figure}

\begin{figure}[t]
 % \vspace{-90pt}
  \centering
  \includegraphics[width=0.8\textwidth]{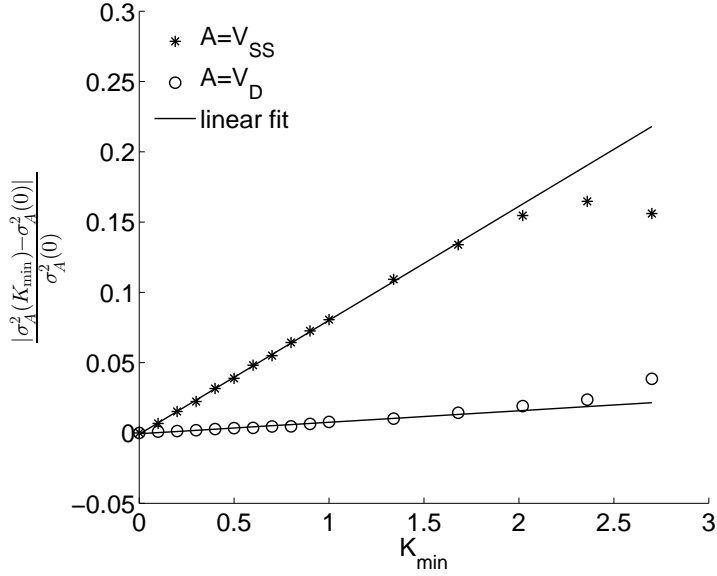}
 % \vspace{90pt}
  \caption{Relative variation in the estimated variances $\sigma_A^2(\er)$ with respect to the reference variances $\sigma_A^2(0)$.}
  \label{fig:comparisonRelDifferenceVariance_SS_DD left}
	\end{figure}

\begin{figure}[t]
 % \vspace{-90pt}
  \centering
  \includegraphics[width=0.8\textwidth]{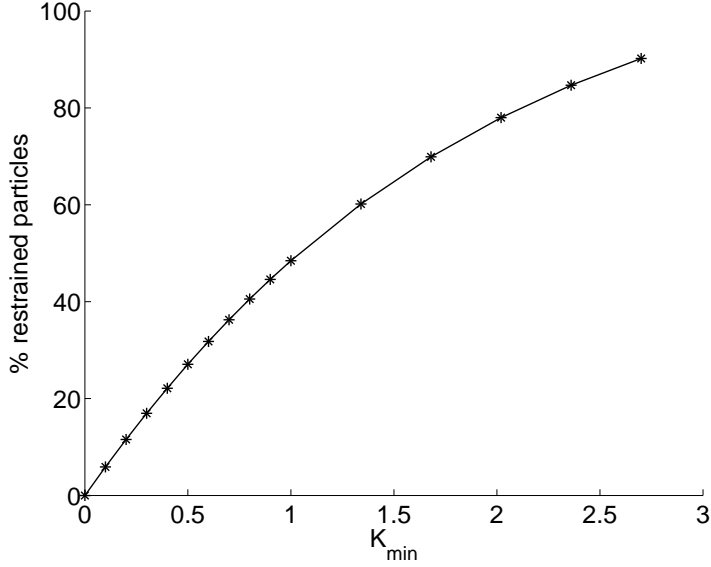}
 % \vspace{90pt}
  \caption{Percentage of restrained particles as a function of $\er$.}
  \label{fig:comparisonRelDifferenceVariance_SS_DD right}
	\end{figure}

%%%%% original format figures
%%%%\begin{figure}[t]
  %%%%\centering
  %%%%%%\vspace{-100pt}
  %%%%\begin{subfigure}[b]{0.5\textwidth}
    %%%%\includegraphics[width=\textwidth]{Figures/comparisonRelDifferenceVariance_SS_DD.eps}
  %%%%\end{subfigure}%
  %%%%~       				%\hspace{-20pt}
  %%%%\begin{subfigure}[b]{0.5\textwidth}
    %%%%%		%\vspace{-220pt}
    %%%%\includegraphics[width=\textwidth]{Figures/PercRestrPart_overKmin.eps}%{Figures/varianceOverPercRestrPart_DD.eps}
  %%%%\end{subfigure}
  %%%%\vspace{50pt}
  %%%%\caption{(Left) Relative variation in the estimated variances $\sigma_A^2(\er)$ with respect to the reference variances $\sigma_A^2(0)$. (Right) Percentage of restrained particles as a function of $\er$.}
  %%%%\label{fig:comparisonRelDifferenceVariance_SS_DD}
%%%%\end{figure}

%-----------------------------------------------------------------------
%                           PROOFS
%-----------------------------------------------------------------------
\section{Proofs of the results}
\label{section proofs}

\subsection{Proof of Lemma~\ref{lemma Lyapunov}}
\label{proof of the lyapunov condition}

The modified Langevin equation can be written as a perturbation of the Langevin equation, namely
\begin{equation}
\label{perturbed Langevin proofs}
\left\{
\begin{aligned}
dq_t &=\left( M^{-1}p_t-\Ze(p_t)\right)dt, \\
dp_t &= -\n V(q_t)dt-\gamma \left(M^{-1}p_t-\Ze(p_t)\right)dt + \sqrt{\frac{2\gamma}{\beta}}\,dW_t,
\end{aligned}
\right.
\end{equation}
where $\Ze(p):=  \n U_{\rm{std}}(p)-\n \U(p) = M^{-1}p-\n\U(p)$ is uniformly bounded as $|\Ze(p)| \leq \ConstNablaU$ in view of Assumption~\ref{hyp mod Langevin ergodicity}. By a direct integration in time of the momenta dynamics,
\begin{equation}
\label{eq:integration_p_perturbation_dynamics}
p_t=\mathrm{e}^{-\gamma t}p_0 + \mathcal{F}_{t} + \mathcal{G}_{t},
\qquad
\mathcal{F}_{t}=\int_{0}^t \Big(-\n V(q_s)+\gamma\Ze(p_s)\Big)\mathrm{e}^{-\gamma (t-s)} \, ds,
\end{equation}
where
\[
\mathcal{G}_{t}=\sqrt{\frac{2\gamma}{\beta}}\int_0^t \mathrm{e}^{-\gamma (t-s)}dW_s
\]
is a Gaussian random variable with mean zero and covariance $\left(1-\mathrm{e}^{-2\gamma t}\right)\beta^{-1}$. Note also that $\mathcal{F}_{t}$ is uniformly bounded; more precisely, $\left|\mathcal{F}_{t}\right|\leq \left\|\n V\right\|_{L^{\infty}}/\gamma+\ConstNablaU$.

Let us first consider the case $s=1$. We introduce $\alpha_t := \mathrm{e}^{-\gamma t}<1$ for a given time $t>0$. With this notation,
\[
\begin{aligned}
\left|p_t\right|^2 = \left|\alpha_t p_0 + \mathcal{F}_{t} + \mathcal{G}_{t}\right|^2 &= \alpha_t^2 \left|p_t\right|^2 + 2\alpha_t p_t^T (\mathcal{F}_t+ \mathcal{G}_t) + \left|\mathcal{F}_t\right|^2 + 2\mathcal{F}_t\mathcal{G}_t + \left|\mathcal{G}_t\right|^2 \\
&\leq \alpha_t^2(1+\e) \left|p_t\right|^2 +\left(2+\frac{1}{4\e}\right)\mathcal{F}_t^2 +\left|\mathcal{G}_t\right|^2+2\alpha_t p_t^T \mathcal{G}_t,
\end{aligned}
\]
where we used Young's inequality to obtain the last line, with a constant $\e>0$ sufficiently small so that $\alpha_t^2(1+\e)<1$. We next take the expectation of the previous inequality, conditionally to the filtration of events up to time~$t$. Since $\mathbb{E}\left[p_t^T \mathcal{G}_t \left| \right. \mathscr{F}_t\right]=0$, it follows
\[
\mathbb{E}\left[\K_1(q_t,p_t) \Big| \mathscr{F}_t\right]\leq \alpha^2(1+\e)\K_1(q_t,p_t) + R,
\]
for some constant $R>0$. This shows the Lyapunov condition for $n=1$. The higher order conditions ($n>1$) can be proved as in~\cite[Section~5.1.5]{Joubaud}, by noting that $|p_t|^{2s}$ is equal to $\alpha_t^{2s} |p_0|^{2s}$ plus some lower order polynomial in~$p_0$.

\subsection{Proof of Lemma~\ref{lemma minorization condition}}
\label{proof of the minorization condition}

The main idea is, as in~\cite[Section~5.1.5]{Joubaud}, to compare the modified Langevin dynamics to the standard Langevin dynamics with zero forces, for which a minorizing measure $\nu_{p^*, t}$ can be explicitly constructed. From the rewriting~\eqref{perturbed Langevin proofs}, we deduce, in view of the momenta evolution~\eqref{eq:integration_p_perturbation_dynamics},
\[
q_t = q_0 + \int_{0}^t \Big(p_s-\Ze(p_s)\Big) ds = q_0 + \int_{0}^t \mathrm{e}^{-\gamma s}p_0 \, ds + \widetilde{\mathcal{G}}_t + \widetilde{\mathcal{F}}_t,
\]
where periodic boundary conditions are considered, and
\[
\widetilde{\mathcal{F}}_t:= \int_{0}^t \mathcal{F}_{s} \, ds - \int_{0}^t \Ze(p_s) \, ds, \qquad \widetilde{\mathcal{G}}_t = \int_{0}^t \mathcal{G}_{s} \, ds.
\]
Note that $\widetilde{\mathcal{F}}_t$ is bounded as
\[
\left|\widetilde{\mathcal{F}}_t\right| \leq \left(\frac{\left\|\n V\right\|_{L^{\infty}}}{\gamma}+ 2\ConstNablaU\right)t,
\]
whereas $\widetilde{\mathcal{G}}_t$ is a Gaussian random variable, which is correlated to $\mathcal{G}_t$. A simple computation shows that
\[
\mathcal{V}:=\mathbb{E}\left[(\widetilde{\mathcal{G}}_t,\mathcal{G}_t)^{T}(\widetilde{\mathcal{G}}_t,\mathcal{G}_t)\right]=
\left(\begin{matrix}
\mathcal{V}_{qq} & \mathcal{V}_{qp} \\
\mathcal{V}_{pq} & \mathcal{V}_{pp}
\end{matrix}
\right),
\quad \qquad
\left\{
\begin{aligned}
\mathcal{V}_{qq}&=\frac{1}{\beta\gamma}\left(2t-\frac{1}{\gamma}\left(3-4\alpha_t+\alpha_t^2\right)\right),\\
\mathcal{V}_{qp} &= \frac{1}{\beta\gamma}\left(1-\alpha_t\right)^2, \\
\mathcal{V}_{pp}&= \frac{1}{\beta}\left(1-\alpha_t^2\right),
\end{aligned}
\right.
\]
where $\alpha_t = \mathrm{e}^{-\gamma t}$ is the same constant as in Section~\ref{proof of the lyapunov condition}. Therefore, for a given measurable set $B \in \mathscr{B}(\mathcal{E})$,
\begin{equation}
\label{eq:minorization_reformulated_with_Gaussians}
\mathbb{P}\left((q_{t},p_{t})\in B \, \Big| \, \left|p_{0}\right| \leq p_{*} \right) \geq \mathbb{P}\left( \left(\widetilde{\mathcal{G}}_t,\mathcal{G}_t\right)\in B - \left(\mathcal{Q}_t,\mathcal{P}_t\right) \, \Big| \, \left|p_{0}\right| \leq p_{*} \right),
\end{equation}
where
\[
\mathcal{Q}_t := q_{0}+ \frac{1-\alpha_t}{\gamma}p_0+\widetilde{\mathcal{F}}_t,
\qquad
\mathcal{P}_t := \alpha_t p_0 + \mathcal{F}_t,
\]
are both bounded by some constant $R > 0$ (depending on $p^*$ and~$t$) when $|p_0| \leq p^*$. Note that there is an inequality in~\eqref{eq:minorization_reformulated_with_Gaussians} since we neglect in fact the periodic images of $q_t$ when writing it as $\mathcal{Q}_t + \widetilde{\mathcal{G}}_t$, the latter two quantities being interpreted as elements of $\R^d$.
Since the matrix $\mathcal{V}$ is definite positive, we can finally consider the following minorizing measure:
\[
\nu_{p^*, t}(B):=Z_R^{-1}\inf_{\left|\mathcal{Q}\right|,\left|\mathcal{P}\right|\leq R}\int_{B - (\mathcal{Q},\mathcal{P})} \exp\left(-\frac{x^T\mathcal{V}^{-1}x}{2}\right) dx,
\]
where $Z_R>0$ is a normalization constant. The proof is concluded by defining $\kappa =(2\pi)^{-d}\det\left(\mathcal{V}\right)^{-1/2}Z_R$.

%---------------------- PREUVES TALAY/KOPEC -----------------------
\subsection{Proof of Lemma~\ref{main theorem kopec}}
\label{proof proposition kopec}

\subsubsection{General structure of the proof}

The proof follows the strategy of~\cite[Proposition A.1]{Kopec}. We recall in this section the general outline of this proof, and highlight the required extensions. The proofs of these extensions are then provided in Section~\ref{Proof of lemma in sketch of the proof}. Without restriction of generality, and in order to simplify the notation, we assume that $A=\Pi_{\mu}A$. We introduce weight functions
\[
\pi_s(p):=\frac{1}{\K_s(p)},
\]
where the Lyapunov functions~$\K_s$ are defined in~\eqref{eq: Lyapunov function}.
We also define
\[%begin{equation}
u(t,q,p) = \left( \mathrm{e}^{t\LopU}A \right)(q,p) = \mathbb{E}\left[A(q_t,p_t) \left| \right. (q_0,p_0)=(q,p)\right].
\]%end{equation}
The following result, central in this proof, gives estimates on derivatives of $u(t)$ in the weighted spaces $L^2(\pi_s)$ (see Section~\ref{Proof of lemma in sketch of the proof} for the proof).

\begin{lemma}
\label{proposition Kopec 1}
Suppose that Assumptions~\ref{hyp mod Langevin ergodicity} and~\ref{hyp THM} hold. For any $n\geq 1$, there exists $\lambda_n>0$ and $s_n \in \mathbb{N}$ such that, for $s \geq s_n$ and $\displaystyle G_{\nu}\leq\rho_s$ with $\rho_s$ sufficiently small, there is $r\in\mathbb{N}$ and $C>0$ for which
\begin{equation}
\forall \left|k\right|\leq n, \qquad
\int_{\mathcal{E}}\left|\du^ku(t,q,p)\right|^2\pi_s(p) \, dp \, dq\leq  C \left\|A\right\|_{W^{n,\infty}_{ \K_{r}}}^2\mathrm{exp}(-\lambda t).
\label{eq:estimate in weighted norm}
\end{equation}
\end{lemma}

Assume in the sequel that $\displaystyle G_{\nu}\leq\rho_s$ for $s$ sufficiently large. In view of the estimates~\eqref{eq:estimate in weighted norm}, and using the fact that $\du^j \pi_s(p)=\psi_{j,s}(p)\pi_s(p)$ with $\psi_{j,s}(p) \to 0$ as $|p| \to +\infty$, we obtain that, for any $n\geq 1$, there exist $s_n \in \mathbb{N}$ such that, for $s\geq s_n$, it is possible to find $r \in \mathbb{N}$ and $C>0$ for which
\[
\forall \left|k\right|+\left|\ell\right|\leq n, \quad \forall t \geq 0,
\qquad
\int_{\mathcal{E}}\left|\du^\ell\Big(\du^k u(t,q,p)\pi_s(p)\Big)\right|^2 dp \, dq\leq C \left\|A\right\|_{W^{\tilde{n},\infty}_{ \K_{r}}}^2\mathrm{exp}(-\lambda t).
%\label{eq:estimate in weighted norm with pi}
\]
By the Sobolev embedding theorem, we can conclude that, for any $n\geq 1$, there exist $s_n,\widetilde{n} \in \mathbb{N}$ such that, for $s\geq s_n$ and provided $G_\nu \leq \rho_s$, it is possible to find $r \in \mathbb{N}$ and $C > 0$ for which
\[
\forall |k| \leq n,
\qquad
\left|\du^ku(t,q,p)\right|\pi_s(p) \leq C \left\|A\right\|_{W^{\widetilde{n},\infty}_{ \K_{r}}}^2 \mathrm{exp}(-\lambda t).
\]
This concludes the proof of Lemma~\ref{main theorem kopec}.

\subsubsection{Proof of Lemma \ref{proposition Kopec 1}}
\label{Proof of lemma in sketch of the proof}

The main tool in the proof of Lemma~\ref{proposition Kopec 1} is the following estimate, which is the counterpart of~\cite[Lemma A.6]{Kopec} for our modified Langevin dynamics.

\begin{lemma}
\label{lemma A.6}
Let $\mathcal{A}$ be a linear operator. Assume that $U\in \mathscr{S}$ and $\de U \in L^{\infty}$. There exists an integer $s_*$ such that, for all $s\geq s_*$, there is a constant $\omega_s>0$ for which the following inequality holds true for any $\zeta,T > 0$:
\begin{equation}
\begin{aligned}
\ds
&\mathrm{exp}(\zeta T)\int_{\mathcal{E}} \left|\mathcal{A} u(t)\right|^2 \pi_s \, dq \, dp + \frac{2\gamma}{\beta}\int_{0}^T\exp(\zeta t) \left(\int_{\mathcal{E}}\left|\n_p \mathcal{A} u(t)\right|^2 \pi_s \, dq \, dp \right) dt   \\
\ds& \qquad \leq \int_{\mathcal{E}}\left| \mathcal{A} u(0)\right|^2 \pi_s \, dq \, dp +  \left(\omega_s +\gamma\left\|\de U\right\|_{L^{\infty}}+\zeta\right)\int_0^T \exp(\zeta t) \left(\int_{\mathcal{E}}\left|\mathcal{A}u(t)\right|^2\pi_{s} \, dp \, dq \right)dt\\
\ds& \qquad \qquad +2\int_{0}^T \exp (\zeta t) \left( \int_{\mathcal{E}}  \left[\mathcal{A},\LopU \right]u(t) \ \mathcal{A}u(t) \pi_s \, dq \, dp \right)dt \ .
\label{eq: lemma A.6}
\end{aligned}
\end{equation}
\end{lemma}

In fact, a careful inspection of the proof shows that, since $U\in\mathscr{S}$, it is possible to avoid the assumption $\de U\in L^{\infty}$ by appropriately increasing the Lyapunov index~$s$. Since $\de U\in L^{\infty}$ for AR-Langevin dynamics, we however keep this assumption.

\begin{proof}
A simple computation shows that
\begin{equation}
2 \LopU \mathcal{A}u(t)\, \mathcal{A}u(t)=\LopU\left( \left|\mathcal{A}u(t)\right|^2\right) - \frac{2\gamma}{\beta}\left|\n_p \mathcal{A}u(t)\right|^2\,.
\label{eq: lemma 1}
\end{equation}
The formal adjoint of the operator $\Lop$ in $L^2(\El)$ is given by
\[
\LopU^{\dagger}=-\left(\n \U \cdot\n_q - \n V\cdot\n_p\right) +\gamma \left(\Delta \U +\n \U\cdot\n_p + \frac{1}{\beta}\Delta_p\right).
\]
In view of Assumption~\ref{hyp mod Langevin ergodicity}, there exists therefore $\omega_s > 0$ such that
\begin{align}
\LopU ^{\dagger}\pi_s &=\left[-\left(\n V+\gamma\n \U\right)\cdot \frac{\n \K_s}{\K_s} +\gamma \Delta \U  - \frac{\gamma}{\beta}\frac{\Delta\K_s}{\K_s} +\frac{2\gamma}{\beta}\frac{\left|\n \K_s\right|^2}{\K_s^2}\right]\pi_s \nonumber\\
&\leq  \left(\omega_s +\gamma\left\|\de U\right\|_{L^{\infty}}\right)\pi_{s}.
\label{eq: estimation adjoint}
\end{align}
With this estimation, we can follow exactly the proof of~\cite[Lemma A.6]{Kopec}, \emph{i.e.} write the expression for $\frac{d}{dt}\left[\mathrm{exp}(\zeta t) \left|\mathcal{A}u(t)\right|^2\right]$, use~\eqref{eq: lemma 1}, integrate the resulting expression in time and with respect to $\pi_s \, dp \, dq$ (for $s$ sufficiently large), and finally use~\eqref{eq: estimation adjoint} to deduce~\eqref{eq: lemma A.6}.
\end{proof}

Let us now prove Lemma~\ref{proposition Kopec 1}. The complete proof is done by induction on $n$. We provide here the complete proofs for $n=0$ and $n=1$, and only sketch the extension to higher orders of derivation since the proof follows the same lines as in~\cite[Appendix~A]{Kopec}.

\paragraph{Case $n=0$.}
Recall first that, in view of Assumption~\ref{hyp mod Langevin ergodicity}, the exponential convergence of the law provided by Theorem~\ref{thm uniform convergence} holds. Denote by $\lambda_\ell$ the corresponding exponential rate of decay for a given $\ell \in \mathbb{N}^*$. For any $r \in \mathbb{N}$, we directly obtain the following decay estimates in $L^{2}(\pi_l)$ when $l> 2r+d/2$: there exists $\widetilde{C}_{l,r}> 0$ such that
\[
\int_{\mathcal{E}}\left|u(t)\right|^2\pi_l\leq \widetilde{C}_{l,r} \, \mathrm{e}^{-2\lambda_r t} \left\|A\right\|_{L^{\infty}_{\K_{r}}}^2.
%\label{eq: m=0}
\]
Note that this corresponds to the case $n=0$ in Lemma~\ref{proposition Kopec 1}.

\paragraph{Case $n=1$.}
We now prove the estimates in the case $n=1$. We first apply Lemma~\ref{lemma A.6} with $\mathcal{A}=\mr{Id}$: there exists $s_* \in \mathbb{N}$ such that, for all $s \geq s_*$ and $ \zeta<2\lambda_r$, there is $C>0$ and $r\in\mathbb{N}$ for which
\begin{equation}
\forall T \in \mathbb{R}_+, \qquad \int_0^T\exp(\zeta t)\left[\int_{\mathcal{E}}\left|\n_p u(t,q,p)\right|^2\pi_s(p) \, dq \, dp \right]dt\leq C \left\|A\right\|_{L^{\infty}_{ \K_{r}}}^2.
\label{eq: estimation dp}
\end{equation}
In order to control derivatives in~$q$, the key idea, going back to~\cite{Talay}, is to use mixed derivatives $\alpha \du_{p_i}-\du_{q_i}$ (for some parameter $\alpha>0$). This allows indeed to retrieve some dissipation in the $q$ direction when $\nabla^2 U$ is positive definite. The next lemma is the most important part of our proof since we show how to extend the use of mixed derivatives to the case when $\nabla^2 U$ is not positive definite.

\begin{lemma}
\label{lemma estimation operator A1}
Consider the operator $L_\alpha := \alpha \n_p-\n_q$ for some parameter $\alpha \in \mathbb{R}$. There exists $s_* \in \mathbb{N}$ such that, for $s \geq s_*$ and provided $G_\nu \leq \rho_s$ (for some constant $\rho_s>0$ defined in~\eqref{eq: def rho m} below), there is $r\in \mathbb{N}$, $\zeta<2\lambda_r$, $\alpha>0$ and $C>0$ for which
\begin{equation}
\forall T > 0, \
\ds\int_0^T\exp(\zeta t) \left[ \int_{\mathcal{E}} \left( \left|L_\alpha u(t,q,p)\right|^2 + \left|\n_pL_\alpha u(t,q,p)\right|^2\right) \pi_s(p)\,dq \, dp \right]dt \leq C \left\|A\right\|_{W^{1,\infty}_{ \K_{r}}}^2.
\label{eq: estimation operator L}
\end{equation}
\end{lemma}

\begin{proof}
Define $L_{\alpha,i}:=\alpha \du_{p_i}-\du_{q_i}$ for $i\in \left\{1,\ldots, d\right\}$. The commutator of $L_{\alpha,i}$ and $\LopU$ is
\[
\begin{aligned}
\left[L_{\alpha,i}, \LopU \right]&=-\alpha \left(\n_p \du_{p_i}\U\right)\cdot \left(\gamma \n_p-\n_q\right)  + \left(\n_q \du_{q_i}V\right)\cdot \n_p \\
&=-\alpha \left(\n_p \du_{p_i}\U\right)\cdot L_\alpha -\alpha \left(\gamma-\alpha\right)\left(\n_p \du_{p_i}\U\right) \cdot  \n_p  + \left(\n_q \du_{q_i}V\right)\cdot\n_p \ . \label{eq: [L_i,Lop]}
\end{aligned}
\]
Introducing $C_V:=\sup_{i,j=1,\ldots,d}\left\|\du_{q_iq_j}^2V\right\|_{L^{\infty}}$, a simple computation shows that
\[
\begin{aligned}
 &2 \sum_{i=1}^d L_{\alpha,i} u(t)[L_{\alpha,i}, \LopU ]u(t) =\\
&\quad=  -2\alpha \sum_{i=1}^d  \sum_{j=1}^{d}	L_iu(t)\left(\du_{p_j} \du_{p_i}U\right) L_{\alpha,j} u(t)
\label{eq: computation [L,L] new} \\
& \qquad + 2\alpha(\alpha-\gamma) \sum_{i=1}^d L_{\alpha,i} u(t) \left(\n \du_{p_i} U\right)\cdot \n_p u(t)
+ 2\sum_{i=1}^dL_{\alpha,i} u(t) \n\left( \du_{q_i} V\right)\cdot \n_p u(t) \\
& \quad\leq  \left( \e_1+\e_2-2\nu\alpha\right) \left|L_\alpha u(t)\right|^2
- 2\alpha  \sum_{i=1}^d\sum_{j=1}^{d} L_{\alpha,i}u(t) \du_{p_j} \du_{p_i}(U -  U_{\nu}) L_{\alpha,j} u(t) \\
& \qquad +\frac{\alpha^2 (\gamma-\alpha)^2}{\e_1}\left( \sup_{\left|j\right|=2}\left\|\du^j U\right\|_{L^{\infty}}^2 \right) \left|\n_p u(t)\right|^2 + \frac{C_V^2}{\e_2}\left|\n_p u(t)\right|^2
\end{aligned}
\]
for any $\e_1, \e_2>0$. With this preliminary computation, we can now choose $\mathcal{A}=L_{\alpha,i}$ in Lemma~\ref{lemma A.6} and sum over $i=1,\dots,d$: for $s \geq s_*$ with $s_*$ sufficiently large,
\begin{align}
& \exp(\zeta T)\int_{\mathcal{E}}\left|L_{\alpha} u(t)\right|^2 \pi_s \, dq \, dp + \frac{2\gamma}{\beta}\int_{0}^T\exp(\zeta t) \left( \int_{\mathcal{E}}\sum_{i=1}^d\left|\n_p L_{\alpha,i} u(t)\right|^2 \pi_s \, dq \, dp \right) dt  \nonumber \\
&\quad \leq\int_{\mathcal{E}} \left| L_{\alpha} u(0)\right|^2 \pi_s \, dq \, dp \nonumber \\
&\quad   +\left(\omega_s +\gamma\left\|\de U\right\|_{L^{\infty}}^2+\zeta+\e_1+\e_2-2\nu\alpha\right)\int_{0}^T\exp(\zeta t) \left( \int_{\mathcal{E}} \left| L_{\alpha} u(t)\right|^2 \pi_s \, dq \, dp \right) dt  \label{eq: ipp 2} \\
&\quad   -2\alpha\int_{0}^T\exp(\zeta t) \left( \int_{\mathcal{E}} L_\alpha u(t)^T \Big[\nabla^2 (U -  U_{\nu})\Big] L_\alpha u(t) \, \pi_s \, dq \, dp \right)dt \label{eq: ipp 3}\\
&\quad   +\frac{\alpha^2 (\gamma-\alpha)^2}{\e_1} \left(\sup_{\left|j\right|=2}\left\|\du^jU\right\|_{L^{\infty}}^2\right)\int_{0}^T\exp(\zeta t) \left( \int_{\El} \left|\n_pu(t)\right|^2 \pi_s \, dq \, dp \right) dt
\nonumber \\
&\quad   +\frac{C_V^2}{\e_2}\int_{0}^T\exp(\zeta t) \left( \int_{\El}\left|\n_pu(t)\right|^2 \pi_s \, dq \, dp \right) dt \nonumber \, .
\end{align}
Since $L_\alpha u(0) = (\alpha \n_p-\n_q) A\in \widetilde{\mathscr{W}}^{1,\infty}_{\K_{r}}$ for some integer $r \leq s_*$ (upon increasing $s_*$), and in view of~\eqref{eq: estimation dp}, the first and the two last terms of the right hand side of the above inequality can be controlled uniformly in time for $\zeta < 2\lambda_r$.

%\begin{remark}
%In the proof of \cite{Kopec}, the term (\ref{eq: ipp 3}) is zero and for the term (\ref{eq: ipp 2}) there exists an $\alpha$ such that it is negative, since $\nu=1/ \max_{i=1,\ldots, N }m_i$.
%\end{remark}
It remains to take care of the terms~\eqref{eq: ipp 2} and~\eqref{eq: ipp 3}. Our strategy is to prove that they are negative when $\zeta< 2\lambda_r$, and can hence be transfered to the lef-hand side of the inequality. To simplify the notation, we denote $\widetilde{U}:=\U- U_{\nu}$. Recall that, by Assumption~\ref{condition dnU-dnUnu}, it holds $\|\n\widetilde{U}\|_{L^{\infty}}\leq G_{\nu}$.  An integration by parts shows that
\[
\begin{aligned}
&-\int_{\mathcal{E}}  \sum_{j=1}^d L_{\alpha,i}u(t)\, \left(\du_{p_j}\du_{ p_i} \widetilde{U}\right) L_{\alpha,j}u(t)\,\pi_s \, dq \, dp
=\int_{\mathcal{E}}\left( \nabla \widetilde{U}\cdot L_\alpha u(t) \right) \mathrm{div}_p\left(L_\alpha u(t)\right) \, \pi_s \, dq \, dp \\
&\quad+\int_{\mathcal{E}} \sum_{j=1}^d \left(\du_{p_i} \widetilde{U}\right) L_{\alpha,j}u(t) \Big(\du_{p_j} L_{\alpha,i}u(t)\Big) \pi_s \, dq \, dp +\int_{\mathcal{E}}\sum_{j=1}^d \left( \nabla \widetilde{U} \cdot L_\alpha u(t)\right)\Big( \nabla \pi_s \cdot L_\alpha u(t)\Big) dq\, dp.
\end{aligned}
\]
With this expression we now estimate the term (\ref{eq: ipp 3}) by
%--------- detail des calculs -----------
% for the first term: use |div(Lu)| \leq \sum_j |\nabla_p L_i u|
% for the second term: do the sum over i, bound by |\nabla U| \sum_j |\partial_j Lu| |L_j u| \leq |\nabla U| |Lu| \sqrt{ \sum_j |\partial_j Lu|^2 } and conclude with \sqrt{ \sum_j |\partial_j Lu|^2 } = \sqrt{ \sum_i,j |\partial_j L_i u|^2 } = \sqrt{\sum_i a_i} \leq \sum_i \sqrt{a_i} = \sum_i \sqrt{ \sum_j |\partial_j L_i u|^2 } = \sum_i |\nabla L_i u|
% for the third one: |\nabla U| |\nabla \pi_s| |Lu|^2
\[
\begin{aligned}
&-2\alpha\int_{0}^T\mathrm{exp}(\zeta t)\left(\int_{\mathcal{E}}\sum_{i=1}^d\sum_{j=1}^d L_iu(t)\left(\du_{p_j}\du_{ p_i} \widetilde{U}\right) L_ju(t)\, \pi_s \, dq \, dp \right) dt \\
&\qquad\qquad \leq 2(1+\mathscr{G}_s)\alpha\e_3\left\| \n\widetilde{U}\right\|_{L^{\infty}}\int_{0}^T\mathrm{exp}(\zeta t) \left(\int_{\mathcal{E}}\left|Lu(t)\right|^2\pi_s \, dq \, dp \right)dt\\
&\qquad \qquad \ \ \ +\frac{2\alpha}{\e_3}\left\| \n\widetilde{U}\right\|_{L^{\infty}}\int_{0}^T\mathrm{exp}(\zeta t)\left(\int_{\mathcal{E}}\sum_{i=1}^d \left|\n_p L_i u(t)\right|^2\pi_s \, dq \, dp \right) dt,
\label{eq: integration par partie}
\end{aligned}
\]
where we have used Young's inequality and introduced a constant $\mathscr{G}_s \in \mathbb{R}_+$ such that $\ds \left|\n_p\pi_s\right|\leq\mathscr{G}_s\pi_s$.

The following conditions are therefore sufficient to ensure that~\eqref{eq: ipp 2} and~\eqref{eq: ipp 3} are non-positive when $\zeta< 2\lambda_r$: there exists $\alpha>0$ such that
\[
\omega_s +\gamma\left\|\de U\right\|_{L^{\infty}}^2+\zeta-2\nu\alpha+2\alpha\e_3(\mathscr{G}_s+1)\left\| \n\widetilde{U}\right\|_{L^{\infty}}<0
%\label{eq: condition estimation alpha 1}
\]
and
\[
\ds\frac{2\gamma}{\beta}>\frac{2\alpha}{\e_3}  \left\| \n\widetilde{U}\right\|_{L^{\infty}}.
%\label{eq: condition estimation alpha 2}
\]
These conditions can be restated as
\[
\frac{\omega_s +\gamma\left\|\de U\right\|_{L^{\infty}}^2+\zeta}{2\left(\nu-\e_3(\mathscr{G}_s+1)\left\| \n\widetilde{U}\right\|_{L^{\infty}}\right)} < \alpha < \frac{\gamma \e_3}{\beta \left\| \n\widetilde{U}\right\|_{L^{\infty}}}.
\]
Since $\zeta$ can be chosen arbitrarily small (while still being positive), the latter condition holds provided $\left\| \n\widetilde{U}\right\|_{L^{\infty}}$:
\[
\left\| \n\widetilde{U}\right\|_{L^{\infty}}<\frac{2\gamma\nu\e_3}{\beta(\omega_s +\gamma\left\|\de U\right\|_{L^{\infty}}^2+\zeta) + 2\gamma(\mathscr{G}_s+1)\e_3^2}.
\]
After optimization with respect to $\e_3$, this leads to the final condition
\[
\left\| \n\widetilde{U}\right\|_{L^{\infty}}<\sqrt{\frac{\nu^2 \gamma}{2\beta(\omega_s +\gamma\left\|\de U\right\|_{L^{\infty}}^2+\zeta)(\mathscr{G}_s+1)}}.
\]
In conclusion, defining
\begin{equation}
\rho_s = \sqrt{\frac{\nu^2 \gamma}{2\beta(\omega_s +\gamma\left\|\de U\right\|_{L^{\infty}}^2)(\mathscr{G}_s+1)}},
\label{eq: def rho m}
\end{equation}
we see that the estimate~\eqref{eq: estimation operator L} holds when the constant $G_{\nu}$ from Assumption~\ref{hyp THM} satisfies $G_\nu \leq \rho_s$.
\end{proof}

The remainder of the proof of Lemma~\ref{proposition Kopec 1} is very similar to the corresponding proof in~\cite{Kopec}. We first combine~\eqref{eq: estimation dp} and Lemma~\ref{lemma estimation operator A1}: there exists $s_* \in \mathbb{N}$ such that for $s \geq s_*$ there exists an integer~$r$, a sufficiently small $\zeta<2\lambda_r$ and $\rho_s > 0$ such that if $G_{\nu}\leq\rho_s$, then there is a constant $C > 0$ for which
\begin{equation}
\forall T \geq 0,
\qquad
\int_0^T\mathrm{exp}(\zeta t)\left(\int_{\mathcal{E}}\left|\n_q u(t)\right|^2 \pi_s \, dq \, dp \right)dt \leq C \left\|A\right\|_{W^{1,\infty}_{ \K_{r}}}^2.
\label{eq: nq}
\end{equation}
We can now again apply Lemma~\ref{lemma A.6}, and sum the estimates obtained with $\mathcal{A}=\du_{p_i}$. Before stating the result, we bound the integrand of the term involving the commutator $\left[\du_{p_i},\LopU \right]$ (for $i=1, \ldots, d$) as:
% uses  ||x||_2 \leq ||x||_1 \leq sqrt{dim}||x||_2
\[
\begin{aligned}
\left| \sum_{i=1}^d\left[\du_{p_i},\LopU \right]u(t)\du_{p_i}u(t) \right| & = \left| \sum_{i=1}^d \nabla_p(\du_{p_i}U) \cdot (\nabla_q -\gamma\nabla_p)u(t) \, \du_{p_i}u(t) \right| \\
& \leq d \left( \sup_{\left|j\right|=2}\left\|\du^j U\right\|_{L^{\infty}}^2 \right) \left|(\nabla_q -\gamma\nabla_p)u(t)\right| \, |\nabla_p u(t)| \\
& \leq d \left( \sup_{\left|j\right|=2}\left\|\du^j U\right\|_{L^{\infty}}^2 \right) \left[\frac12 |\nabla_q u(t)|^2 + \left(\gamma+\frac12\right) |\nabla_p u(t)|^2\right].
\end{aligned}
\]
Then, for $s \geq s_*$ (with $s_*$ sufficiently large) and for all $T \geq 0$,
\[
\begin{aligned}
& \mathrm{exp}(\zeta T)\int_{\mathcal{E}}\left|\n_p u(t)\right|^2 \pi_s \, dq \, dp \\
& \leq \int_{\mathcal{E}}\left| \n_pu(0)\right|^2 \pi_s \, dq \, dp + d\left(\sup_{\left|j\right|=2}\left\|\du^j U\right\|_{L^{\infty}}\right)\int_{0}^T\mathrm{exp}(\zeta t)\left(\int_{\mathcal{E}} \left|\n_{q}u(t)\right|^2 \pi_s \, dq \, dp \right) dt \\
& \ \ \ +  \left(\omega_s +\zeta + \sup_{\left|j\right|=2} \left\|\du^j U\right\|_{L^{\infty}} \left(d+2\gamma d+\gamma\right)\right)\int_0^T  \mathrm{exp}(\zeta t) \left(\int_{\mathcal{E}}\left|\n_pu(t)\right|^2 \pi_{s} \, dq\,dp\right) dt.
\end{aligned}
\]
In view of~\eqref{eq: estimation dp} and~\eqref{eq: nq} and since $\n_pu(0)=\n_pA\in \widetilde{W}^{1,\infty}_{\K_{\widetilde{r}}}$ for some integer $\widetilde{r}\in\mathbb{N}$, we see that there exists $s_* \geq 1$ sufficiently large such that, for any $s \geq s_*$ and $\zeta>0$ sufficiently small, and provided $G_\nu \leq \rho_s$, there is a constant $C > 0$ and an integer~$r$ for which
\[
\int_{\mathcal{E}}\left|\n_p u(t)\right|^2 \pi_{s}\, dq \, dp \leq C \left\| A\right\|_{W^{1,\infty}_{\K_r}}^2 \mathrm{exp}(-\zeta T).
\]
To conclude to Lemma~\ref{proposition Kopec 1} for $n=1$, it remains to apply Lemma \ref{lemma A.6} with $\mathcal{A}=\n_q$ in order to obtain an estimate similar to the one above, but for $\left|\n_q u(t)\right|^2$. This is possible in view of the following bounds on the commutator: for all $i=1, \ldots, d$,
\[
\Big| \left[\du_{q_i},\LopU \right]u(t)\du_{q_i}u(t) \Big| = \left| \n \left(\du_{q_i}V\right)\cdot \n_p u(t)\du_{q_i}u(t) \right| \leq \left|\n_{p}u(t)\right|^2+C_V^2\left|\du_{q_i}u(t)\right|^2.
\]

\paragraph{General $n$.}
The remainder of the proof is done by induction of $n$ and relies on the control of the commutators  $\left[\du^k_q,\LopU \right]$ with $\left|k\right|=n$, which are independent of $\U $, as well as
\[
\left|\left[\du^k_p,\LopU \right]\psi\right|\leq \sum_{\substack{i\in \mathbb{N}^{2d}\\ \left|i\right|\leq n}}\mathcal{P}_i\left|\du^i\psi\right|,
%\label{eq: induction [Dk,L]}
\]
where $\mathcal{P}_i$ are positive polynomial functions that depend on the polynomial growth of $\U$ and its derivatives. These polynomial functions can be controlled with Lyapunov weights for sufficiently large indices. In addition, the same approach as in the proof of Lemma~\ref{lemma estimation operator A1} is used to estimate the extra term arising from missing positivity of $\nabla^2 U$, namely
\[
\begin{aligned}
&-2\alpha\int_{0}^T\mathrm{exp}(\zeta t)\left(\int_{\mathcal{E}}\sum_{i=1}^d\sum_{j=1}^d L_{\alpha,i}\left(\du^n_qu(t)\right) \left[ \du_{p_ip_j}^2 \left(\U -  U_{\nu}\right) \right]L_{\alpha,j}\left(\du^n_qu(t)\right)\pi_s \, dq \, dp \right) dt  \\
&\qquad\leq 2(1+\mathscr{G}_s)\alpha\e_3\left\| \n\widetilde{U}\right\|_{L^{\infty}}\int_{0}^T\mathrm{exp}(\zeta t)\left(\int_{\mathcal{E}}\left|L_\alpha\left(\du^n_qu(t)\right)\right|^2\pi_s \, dq \, dp \right) dt \\
&\qquad \ \ \ +\frac{2\alpha}{\e_3}\left\| \n\widetilde{U}\right\|_{L^{\infty}}\int_{0}^T\mathrm{exp}(\zeta t) \left(\int_{\mathcal{E}}\left|\n_p L_\alpha\left(\du^n_qu(t)\right)\right|^2\pi_s \, dq \, dp \right) dt.
\end{aligned}
\]
Therefore, the result is obtained when the same condition~\eqref{eq: def rho m} on $G_{\nu}$ is satisfied. Note however that this condition depends on~$s$, hence on~$n$ since $s$ has to be larger than some index~$s_n$.

%----------------------------- VARIANCE PERTURBATION --------------------
\subsection{Proof of Proposition \ref{proposition variance perturbation}}
\label{proof of the proposition variance perturbation}

\subsubsection{General structure of the proof}

We define the AR perturbation function as
\[
D_{\er}(p):=\n U_{0, \ef}(p)-\n U_{\er, \ef}(p)\,.
%\label{eq: definition De}
\]
This allows to write the generator $\Lopgen_{\er,\ef}$ of the AR-Langevin dynamics~\eqref{eq: generator modified Langevin} as a perturbation of the generator $\Lopgen_{0, \ef}$:
\[
\Lopgen_{\er, \ef} = \Lopgen_{0, \ef} - D_{\er}(p)\cdot \Lpert, \qquad \Lpert := \n_q-\gamma \n_p.
%\label{Leps}
\]
For notational convenience we omit the subscript $\ef$ and simply write $\Lop_{\er}:=\Lopgen_{\er,\ef}$. We also denote by $\mu_{\er}$ the invariant measure associated with $\Lop_{\er}$, and by $\Pi_{\er}$ the projection
\[
\Pi_{\er}f = f - \int_\El f \, d\mu_{\er}.
\]
%Note that $A\in L^2({\mu_{\er}})$ if and only if $A\in L^2({\mu_{0}})$, but $ \widetilde{L}^2({\mu_{\er}})\neq \widetilde{L}^2({\mu_{0}})$.

For a given observable $A\in\mathscr{S}$, the asymptotic variance associated with the corresponding time averages reads, in view of~\eqref{eq: variance expression}:
\begin{equation}
\ds\sigma^2_A(\er) = -2\int_{\El} \Phi_{A,\er} A \, d\mu_{\er},
\label{eq: variance F}
\end{equation}
where $\Phi_{A,\er}\in \mathscr{S}$ is the unique solution in $L^\infty_{\K_s}$ ($s$ being such that $A \in L^\infty_{\K_s}$) of the following Poisson equation:
\begin{equation}
\Lop_{\er} \Phi_{A,\er}=\Pi_{\er}A, \qquad \Pi_{\er} \Phi_{A,\er} = 0.
\label{eq: poisson F}
\end{equation}
Similarly, the limiting variance for $\er = 0$ can be rewritten as
\begin{equation}
\sigma^2_A(0) = -2\int_{\El} \Phi_{A,0}A \, d\mu_{0},
\qquad
\Lop_0 \Phi_{A,0}=\Pi_{0}A,
\qquad
\Pi_0 \Phi_{A,0} = 0.
\label{eq: original variance}
\end{equation}

In order to prove the convergence of~\eqref{eq: variance F} to~\eqref{eq: original variance} and to identify the linear term in~$\er$, the idea is to expand $\mu_{\er}$ and $\Phi_{A,\er}$ in powers of $\er$. To this end, we rewrite the Poisson equation~\eqref{eq: poisson F} as
\[
\Pi_0 \left(\Pi_0 - \Lop_0^{-1}\Pi_0 D_{\er}\cdot \Lpert\right) \Phi_{A,\er} = \Lop_0 ^{-1}\Pi_0 A.
\]
The operator $\Lop_0^{-1}\Pi_0 D_{\er}\cdot \Lpert$ is not bounded (since $\Lpert$ contains derivatives in~$q$, which cannot be controlled by~$\Lop_0$), so that it is not possible to write the inverse of $\Pi_0 - \Lop_0^{-1}\Pi_0 D_{\er}\cdot \Lpert$ as some Neumann series. It is however possible to consider a pseudo-inverse operator by truncating the Neumann series at order~$n$. This motivates the introduction of the following approximation of the solution of~\eqref{eq: poisson F}:
\[
\Phi^{n}_{A,\er} := \sum_{k=0}^{n}\left(\Lop_0^{-1}\Pi_0D_{\er}\cdot \Lpert\right)^k \Lop_0^{-1}\Pi_0A \ .
%\label{eq: definition Fn}
\]
The corresponding approximation of the variance reads
\begin{equation}
\sigma^2_{A,n}(\er) := -2\int_{\El} \left(\Pi_{\er}\Phi_{A,\er}^n\right)A  \, d\mu_{\er}.
\label{def truncated variance}
\end{equation}
The connection with the exact variance~\eqref{eq: variance F} is given by the following lemma, which is proved in Section~\ref{proof of theorem truncated variance}. We introduce a critical value $\ef^*$ such that Assumption~\ref{hyp THM} is satisfied for $0 \leq \er \leq \ef/2$ and $\ef \leq \ef^*$ (see Section~\ref{proof of lemma arps condition}). This allows to resort to Lemma~\ref{main theorem kopec}.

\begin{lemma}
Fix $0 < \ef \leq \ef^*$. Then, for any $A\in \mathscr{S}$ and for all $n\geq 1$, there exists a constant $C_{A,n} >0$ such that
\[
\forall \ 0\leq\er\leq \frac{\ef}{2},
\qquad
\left|\sigma^2_{A}(\er)-\sigma^{2}_{A,n}(\er)\right|\leq C_{A,n} \er^{n+1}.
\]
\label{theorem truncated variance}
\end{lemma}

The key point in the proof of Lemma~\ref{theorem truncated variance} are the following estimates (see Section \ref{proof of Lemma estimate LopD n} for the proof).

\begin{lemma}
\label{corollary: estimate LopD n}
Fix $0 < \ef \leq \ef^*$ and $A\in \mathscr{S}$. For any $n\geq 1$, there exist $s_n,l_n\in \mathbb{N}$ such that, for any $s\geq s_n$, there is $r_n \in\mathbb{N}$ and $\tilde{C}_n>0$ for which
\[
\forall \, 0 \leq \er \leq \frac{\ef}{2},
\qquad
\left\| \left(\LopZero^{-1}\Pi_0 D_{\er}\cdot \Lpert\right)^n \Pi_0 A \right\|_{L^{\infty}_{\K_{s}}}\leq  \tilde{C}_n \er^n \left\| A\right\|_{W^{l_n,\infty}_{\K_{r_n}}} .
\label{eq: estimate LopD n}
\]
\end{lemma}

Proposition~\ref{proposition variance perturbation} now straightforwardly follows by combining Lemma~\ref{theorem truncated variance} and the following expansion in powers of~$\er$ of the truncated variance (whose proof can be read in Section~\ref{proof of variance perturbation main theorem}).

\begin{proposition}
\label{variance perturbation main theorem}
Fix $0 < \ef \leq \ef^*$. There exists a constant $\mathscr{K}\in \R$ such that, for any $n\geq 1$ and $0\leq \er \leq \ef/2$ sufficiently small,
\[
\sigma_{A,n}^2(\er) = \sigma_A^2(0) + \mathscr{K} \er + \mathrm{O}(\er^2).
\]
\end{proposition}

%------------------------ tech results expansion -----------------
\subsubsection{Technical results on expansions with respect to $\er$}
\label{proof Ueps minus Uzero intro}

Recall that the function $f_{0,\ef}$ (with $f_{0,\ef}$ defined in~\eqref{eq: def fZeroK 1}) belongs to $C^\infty(\mathbb{R},[0,1])$. The next result shows that the same is true for
\[
f_{\er,\ef} = f_{0,\ef} \circ \theta_{\er},
\]
with $\theta_{\er}$ defined in~\eqref{eq: definition theta}. This is not obvious a priori since $\theta_{\er}$ is only piecewise $C^\infty$, with singularities on the first order derivative at~$\er$ and~$\ef$. In fact, it can even be proved that $f_{\er,\ef} - f_{0,\ef}$ and all its derivatives are small when $\er$ is small.

\begin{lemma} For any $0\leq \er <\ef$, the function $f_{\er,\ef}$ belongs to $C^{\infty}(\mathbb{R},[0,1])$. Moreover, its derivatives have a compact support in $\left[0,\ef\right]$. Finally, for any $n_0\in\mathbb{N}$ and $\delta > 0$, there exists a constant $C_{n_0,\ef,\delta}>0$ such that
\begin{equation}
\forall \, 0\leq n\leq n_0, \quad \forall \, \er \in [0,\ef-\delta],
\qquad
\left\|f_{\er,\ef}^{(n)} - f_{0,\ef}^{(n)}\right\|_{L^{\infty}}\leq C_{n_0,\ef,\delta}\er.
\label{eq: corollary ftheta minus f 1}
\end{equation}
\label{lemma Cinfty in p f}
\end{lemma}

\begin{proof}
The function $\theta_{\er}$ is defined piecewise on three intervals $[0,\er)$, $(\er,\ef)$ and $(\ef,+\infty)$. In the interior of each interval, both $f_{0,\ef}$ and $\theta_{\er}$ are $C^\infty$, and so is therefore their composition. In addition, $f_{\er,\ef}$ is constant on $(\ef,+\infty)$, hence all derivatives vanish on this interval. To prove that $f_{\er,\ef}$ is $C^\infty$ with derivatives of compact support, it therefore suffices to prove that all derivatives can be extended by continuity at the points $\er$ and $\ef$.

Since $f_{0,\ef}$ is constant outside the interval $[\er,\ef]$, a simple computation shows that, for $n \geq 1$,
\begin{equation}
\label{eq:derivative_f_Kmin}
\left(f_{0,\ef}\circ\theta_{\er}\right)^{(n)}(x) =
\left\{ \begin{aligned}
0 & \qquad \mathrm{for} \ 0 \leq x < \er, \\
\left(\frac{\ef}{\ef-\er}\right)^nf_{0,\ef}^{(n)}(\theta_{\er}(x))\, & \qquad \mathrm{for} \ \er < x < \ef, \\
0 & \qquad \mathrm{for} \ x > \ef.
\end{aligned} \right.
\end{equation}
It is therefore obvious to check the continuity at $\er$ and $\ef$ since all derivatives of $f_{0,\ef}$ vanish at~0 and~$\ef$, and $\theta_{\er}(\er) = 0$ while $\theta_{\er}(\ef) = \ef$.

Moreover, it is easy to check that $\left| \theta_{\er}(x)-x \right| \leq \er$, so that the estimate~\eqref{eq: corollary ftheta minus f 1} already follows in the case $n=0$ since $f_{0,\ef}$ is Lipschitz continuous. To obtain the same result for higher order derivatives, we note that the $n$-th order derivative can be rewritten as
\[
f_{\er,\ef}^{(n)} = f_{0,\ef}^{(n)} \circ\theta_{\er}+ f_{0,\ef}^{(n)}\circ\theta_{\er}\left(\left[\left(\frac{\ef }{\ef -\er }\right)^n-1\right]\1_{\left[\er,\ef\right]}\right)\,.
%\label{eq: derivation f theta}
\]
Therefore, $f_{\er,\ef}^{(n)} - f_{0,\ef}^{(n)}$ is the sum of (i) $f_{0,\ef}^{(n)} (\theta_{\er}) - f_{0,\ef}^{(n)}$, which is of order~$\er$ in $L^\infty$ norm by the same argument as before since $f_{0,\ef}^{(n)}$ is Lipschitz continuous; and (ii) a remainder term of order~$\er$ since $f_{0,\ef}^{(n)} \circ\theta_{\er}$ is uniformly bounded; while for any $0<\delta<\ef$ there exists $R_{n,\delta}>0$ such that
\[
\forall \, \er \in [0,\ef-\delta],
\qquad
\left|\left(\frac{\ef }{\ef -\er }\right)^n-1\right|\leq R_{n,\delta} \er\,.
\]
This allows to obtain the desired result.
\end{proof}	

In view of the definition~\eqref{definition U}-\eqref{eq: definition zeta} of $U_{\er,\ef}(p)=\sum_{i=1}^Nu_{\er,\ef}(p_i)$, we can deduce the following estimates on $U_{\er,\ef} - U_{0,\ef}$ and its derivatives, which allow in particular to control $D_{\er}$. To state the result, we introduce
\[
\mathscr{C}_{K}:=\left\{p \in \mathbb{R}^d \ \left| \ \forall i=1, \ldots N, \ \frac{p_i^2}{2 m_i}\leq K\right. \right \}\,.
%\label{eq: definition CKmax}
\]

\begin{corollary}
\label{lemma Ueps minus Uzero intro}
For any $0\leq \er<\ef$, the function $U_{\er, \ef}$ belongs to $C^{\infty}$. For any $n\geq 0$ and $\left|\alpha\right|= n$, the function $\du^{\alpha}\left(U_{\er, \ef}-U_{0,\ef}\right)$ has a compact support in $\mathscr{C}_{\ef}$. Moreover, for any $n_0\geq 0$ and $\delta > 0$, there exists a constant $C_{n_0,\delta,\ef}>0$ such that
\begin{equation}
\forall \left|\alpha\right|\leq n_0, \ \forall \, \er \in [0,\ef-\delta],
\quad
\left\|\du^{\alpha} U_{\er, \ef}-\du^{\alpha} U_{0,\ef}\right\|_{L^{\infty}}\leq C_{n_0,\delta,\ef} \er \,.
\label{eq: Ueps minus Uzero intro}
\end{equation}
\label{corollary Ukmin Cinfty}
\end{corollary}

In order to obtain more precise statements about the behavior of the functions $f_{\er,\ef}(x)$ for small values of~$\er$, a natural idea would be to perform Taylor expansions with respect to this parameter. The difficulty is however that the derivatives with respect to $\er$ of the shift function $\theta_{\er}(x)$ are not continuous in $x$. This prevents to write directly remainders of order~$\er^2$. Before stating the precise result in Lemma~\ref{lemma DkimKim conv}, we need another technical ingredient.

\begin{lemma}
\label{lemma KminKmax tau}
Fix $\ef > 0$ and define $\ds \tilde{\tau}(x):=\frac{x-\ef}{\ef}$. Then, for any $n\geq 0$ and $\delta > 0$, there exists $C_{n,\delta}>0$ such that
\begin{equation}
\label{eq: lemma bound fcirctheta minus f}
\forall \, \er \in [0,\ef-\delta],
\qquad
\left\|\frac{f_{0,\ef}^{(n)} \circ \theta_{\er} - f_{0,\ef}^{(n)}  }{\er}-f_{0,\ef}^{(n+1)} \circ \tilde{\tau}\right\|_{L^{\infty}}\leq C_{n,\delta} \, \er\,.
\end{equation}
\end{lemma}

\begin{proof}
Note that, formally, $\tilde{\tau}$ is the derivative of $\theta_{\er}$ on $\left[\er,\ef\right]$ with respect to~$\er$, evaluated at $\er = 0$. Recall also $\left|\theta_{\er}(x)-x\right|\leq \er$. Simple computations show that there exists $C_\delta>0$ such that
\[
\forall x \in \R^{+},
\qquad
\left|\frac{\theta_{\er}(x)-x}{\er} - \tilde{\tau}(x) \right|\leq C_{\delta}\er\,.
\]
Since $f_{0,\er}\in C^{\infty}$, there exists $t\in \left[0,1\right]$ such that
\begin{equation}
\label{eq: alpha tangent}
f_{0,\ef}^{(n)}\big(\theta_{\er}\left(x\right)\big)-f^{(n)}_{0,\ef}\left(x\right) = f^{(n+1)}_{0,\ef}\Big(x+t\big(\theta_{\er}(x)-x\big)\Big) \big(\theta_{\er}(x)-x\big)\,.
\end{equation}
Therefore, for $x \in [0,\ef]$,
\begin{equation}
\label{fnpone minus fn div kmin min fnp1tau}
\begin{aligned}
&\left|\frac{f_{0,\ef}^{(n)}\left(\theta_{\er}\left(x\right)\right)-f_{0,\ef}^{(n)}\left(x\right)}{\er}-f^{(n+1)}_{0,\ef}(x)\tilde{\tau}(x)\right| \\
&=\left |f^{(n+1)}_{0,\ef}\Big(x+t\big(\theta_{\er}(x)-x\big)\Big)\frac{\theta_{\er}(x)-x}{\er}-f^{(n+1)}_{0,\ef}(x)\tilde{\tau}(x)\right| \\
&\leq \left|f^{(n+1)}_{0,\ef}\Big(x+t\big(\theta_{\er}(x)-x\big)\Big) - f_{0,\ef}^{(n+1)}(x)\right| \, \left|\frac{\theta_{\er}(x)-x}{\er}\right| +\left | f^{(n+1)}_{0,\ef}(x)\right| \, \left |\frac{\theta_{\er}(x)-x}{\er} - \tilde{\tau}(x) \right| \\
&\leq \left(\left\| f^{(n+2)}_{0,\ef}\right\|_{L^{\infty}\left(\left[0,\ef\right]\right)}   + C_\delta \left\| f^{(n+1)}_{0,\ef}\right\|_{L^{\infty}\left(\left[0,\ef\right]\right)} \right)\er,
\end{aligned}
\end{equation}
where we have used the following equality: there exists $\alpha \in [0,1]$ such that
\[
f^{(n+1)}_{0,\ef}\Big(x+t\big(\theta_{\er}(x)-x\big)\Big) - f_{0,\ef}^{(n+1)}(x) = t \, f^{(n+2)}_{0,\ef}\Big(x+\alpha\big(\theta_{\er}(x)-x\big)\Big) \big(\theta_{\er}(x)-x\big),
\]
together with the bound $|\theta_{\er}(x)-x| \leq \er$.
\end{proof}

\begin{lemma}
Fix $\ef>0$. There exist functions $\mathscr{D}_i\in C^{\infty}(\R^d)$ (for $i=1, \ldots, N$), with compact support in $\mathscr{C}_{\ef}$, such that, for $0< \delta < \ef$ and $r \in \mathbb{N}$, there is $C_{r,\delta}>0$ such that
\begin{equation}
\er\in \left[0,\ef-\delta\right],\quad\left\|D_{\er,i}-\er\mathscr{D}_i\right\|_{W^{r,\infty}}\leq C_{r,\delta} \er^{2}.
\label{eq: DkimKim esti Kmin2}
\end{equation}
\label{lemma DkimKim conv}
\end{lemma}

\begin{proof}
Recall that the functions $D_{\er,i}: \R^D\rightarrow\R^D$ are defined, for $i=1,\ldots, N$, as
\[
\begin{aligned}
  D_{\er,i}(p)& = \left[f_{\er,\ef}\left(\frac{\left|p_i\right|^2}{2m_i}\right)-f_{0,\ef}\left(\frac{\left|p_i\right|^2}{2m_i}\right)\right] \frac{p_i}{m_i}\\
& \qquad + \frac{\left|p_i\right|^2}{2m_i} \left[f^{\prime}_{\er,\ef}\left(\frac{\left|p_i\right|^2}{2m_i}\right)-f_{0,\ef}^{\prime}\left(\frac{\left|p_i\right|^2}{2m_i}\right)\right]\frac{p_i}{m_i}\,.
\end{aligned}
\]
We next define, for $i=1, \ldots, N$, the function
\[
\mathscr{D}_i(p) := 
\left\{ \begin{aligned}
\left[ f_{0,\ef}^{\prime}\left(\frac{\left|p_i\right|^2}{2m_i}\right)+\frac{\left|p_i\right|^2}{2m_i} f_{0,\ef}^{\prime\prime}\left(\frac{\left|p_i\right|^2}{2m_i}\right)\right]\tilde{\tau}\left(\frac{\left|p_i\right|^2}{2m_i}\right)\frac{p_i}{m_i},  & \qquad \mathrm{for} \ \frac{\left|p_i\right|^2}{2m_i}\in \left[0, \ef  \right], \\
0 ,& \qquad \mathrm{for} \ \frac{\left|p_i\right|^2}{2m_i} \geq \ef,
\end{aligned} \right.
\]
where $\tilde{\tau}$ is defined in Lemma~\ref{lemma KminKmax tau}. Recall that $f_{0,\ef}\in C^{\infty}$ and $f_{0,\ef}^{(n)}$ have compact support on $\left[0,\ef\right]$ for $n\geq 1$. Therefore, $\mathscr{D}_i\in C^{\infty}$ also has compact support in $\mathscr{C}_{\ef}$.

The case $r=0$ of \eqref{eq: DkimKim esti Kmin2} follows directly from Lemma~\ref{lemma KminKmax tau} with $n=0$ and $n=1$. Let us now consider the case $r=1$ more carefully. To simplify the presentation, we consider separately the two terms in the sums defining the functions $D_{\er,i}$ and $\mathscr{D}_i$, \emph{i.e.} $D_{\er,i}= D_{\er,i,1} + D_{\er,i,2}$ and $\mathscr{D}_i = \mathscr{D}_{i,1} + \mathscr{D}_{i,2}$ with
\[
D_{\er,i,1}(p) = \left[f_{\er,\ef}\left(\frac{\left|p_i\right|^2}{2m_i}\right)-f_{0,\ef}\left(\frac{\left|p_i\right|^2}{2m_i}\right)\right] \frac{p_i}{m_i},
\]
and
\[
\mathscr{D}_{i,1}(p) :=
\left\{ \begin{aligned}
f_{0,\ef}^{\prime}\left(\frac{\left|p_i\right|^2}{2m_i}\right) \tilde{\tau}\left(\frac{\left|p_i\right|^2}{2m_i}\right)\frac{p_i}{m_i},  & \qquad \mathrm{for} \ \frac{\left|p_i\right|^2}{2m_i}\in \left[0, \ef  \right], \\
0 ,& \qquad \mathrm{for} \ \frac{\left|p_i\right|^2}{2m_i} \geq \ef.
\end{aligned} \right.
\]
We present the estimates only for the difference $D_{\er,i,1}/\er-\mathscr{D}_{i,1}$ since similar computations allows to control the difference $D_{\er,i,2}/\er-\mathscr{D}_{i,2}$. For $\alpha,\alpha' \in \{1,\dots,D\}$, we denote by $p_{i,\alpha}$ the $\alpha$th component of the momentum of the $i$th particle and by $D_{\er,i,1,\alpha'}$ and $\mathscr{D}_{i,1,\alpha'}$ the $\alpha'$th components of $D_{\er,i,1}$ and $\mathscr{D}_{i,1}$. Then, for $p_i \in \mathscr{C}_{\ef}$,
\[
\begin{aligned}
&\left|\du_{p_{i,\alpha}}\left( \frac{D_{\er,i,1,\alpha'}}{\er} - \mathscr{D}_{i,1,\alpha'} \right)(p)\right|=\\
& \leq\frac{\delta_{\alpha,\alpha'}}{m_i}\left| \frac{\dps f_{0,\ef}\circ \theta_{\er}\left(\frac{\left|p_i\right|^2}{2m_i}\right)-f_{0,\ef}\left(\frac{\left|p_i\right|^2}{2m_i}\right) }{\er} -  f_{0,\ef}^{\prime}\left(\frac{\left|p_i\right|^2}{2m_i}\right)\tilde{\tau}\left(\frac{\left|p_i\right|^2}{2m_i}\right)  \right| \\
& +\frac{|p_{i,\alpha}p_{i,\alpha'}|}{m_i^2}\left|  \frac{\dps f_{0,\ef}^{\prime}\circ \theta_{\er}\left(\frac{\left|p_i\right|^2}{2m_i}\right)\theta_{\er}^{\prime}\left(\frac{\left|p_i\right|^2}{2m_i}\right) -f_{0,\ef}^{\prime}\left(\frac{\left|p_i\right|^2}{2m_i}\right) }{\er} \right. \\
& \left.
\phantom{+ \left|  \frac{\dps f_{0,\ef}^{\prime}\left(\frac{\left|p_i\right|^2}{2m_i}\right)}{\er} \right.} -f_{0,\ef}^{\prime\prime}\left(\frac{\left|p_i\right|^2}{2m_i}\right)\tilde{\tau}\left(\frac{\left|p_i\right|^2}{2m_i}\right) -  f_{0,\ef}^{\prime}\left(\frac{\left|p_i\right|^2}{2m_i}\right)\frac{1}{\ef} \right| \\
& \leq \frac{\delta_{\alpha,\alpha'}}{m_i}\left| \frac{\dps f_{0,\ef}\circ \theta_{\er}\left(\frac{\left|p_i\right|^2}{2m_i}\right)-f_{0,\ef}\left(\frac{\left|p_i\right|^2}{2m_i}\right) }{\er} -  f_{0,\ef}^{\prime}\left(\frac{\left|p_i\right|^2}{2m_i}\right)\tilde{\tau}\left(\frac{\left|p_i\right|^2}{2m_i}\right)  \right| \\
&\quad +\frac{|p_{i,\alpha}p_{i,\alpha'}|}{m_i^2}\left|  \frac{\dps f_{0,\ef}^{\prime}\circ \theta_{\er}\left(\frac{\left|p_i\right|^2}{2m_i}\right)-f_{0,\ef}^{\prime}\left(\frac{\left|p_i\right|^2}{2m_i}\right) }{\er}-   f_{0,\ef}^{\prime\prime}\left(\frac{\left|p_i\right|^2}{2m_i}\right)\tilde{\tau}\left(\frac{\left|p_i\right|^2}{2m_i}\right)\right|\\
&\quad+\frac{|p_{i,\alpha}p_{i,\alpha'}|}{m_i^2}\left|\left[f^{\prime}_{0,\ef}\circ \theta_{\er}\left(\frac{\left|p_i\right|^2}{2m_i}\right)-f^{\prime}_{0,\ef}\left(\frac{\left|p_i\right|^2}{2m_i}\right)\right]\frac{\dps \theta_{\er}^{\prime}\left(\frac{\left|p_i\right|^2}{2m_i}\right)-1}{\er}\right| \\
& \quad + \frac{|p_{i,\alpha}p_{i,\alpha'}|}{m_i^2}\left|f_{0,\ef}^{\prime}\left(\frac{\left|p_i\right|^2}{2m_i}\right)\left[\frac{\dps \theta_{\er}^{\prime}\left(\frac{\left|p_i\right|^2}{2m_i}\right)-1}{\er}-\frac{1}{\ef}\right] \right|,
\end{aligned}
\]
where we used $\tilde{\tau}^{\prime}\left(x\right)=1/\ef$ for $x\in[0,\ef]$. The first two terms in the last inequality can be bounded by $C^*\er$ for some constant $C^* \in \mathbb{R}_+$ in view of Lemma~\ref{lemma KminKmax tau}. For the last two terms, distinguish the cases $p_i \in \mathscr{C}_{i,\er}$ and $p_i \in \mathscr{C}_{i,\ef} \backslash \mathscr{C}_{i,\er}$, where for $K\geq 0$ we define
\[
\mathscr{C}_{i,K}:=\left\{p_i \in \mathbb{R}^D \, \left| \, \frac{\left|p_i\right|^2}{2m_i}\leq K \right. \right\}\,.
\]
When $p_i \in \mathscr{C}_{i,\er}$, the third term disappears since $\theta_{\er}^{\prime}(x)=1$ on $\left[0,\er\right]$. In addition,
\[
\sup_{p_i\in \mathscr{C}_{i,\er}}\frac{\left|p_i\right|^2}{m_i^2}\leq \frac{2\er}{m_i},
\]
so that
\[
\begin{aligned}
&\left\| \du_{p_{i,\alpha}}\left( \frac{D_{\er,i,1,\alpha'}}{\er} - \mathscr{D}_{i,1,\alpha'} \right) \right\|_{L^{\infty}\left(\mathscr{C}_{i,\er}\right)}\leq \left( C^*+\frac{2}{m_i\ef}\left\| f_{0,\ef}^{\prime} \right\|_{L^{\infty}\left(\left[0,\er\right]\right)}\right)\er\,.
\end{aligned}
\]
When $p_i \in \mathscr{C}_{i,\ef}\backslash \mathscr{C}_{i,\er}$, we use $\theta_{\er}^{\prime}(x) = \ef/(\ef-\er)$ for $x\in\left[\er,\ef\right]$, so that there exists $C_\delta > 0$ such that
\begin{equation}
\sup_{x\in \left[\er,\ef\right]}\left |\frac{\theta_{\er}^{\prime}(x)-1}{\er} - \frac{1}{\ef} \right|\leq C_{\delta} \er,
\quad
\sup_{x\in \left[\er,\ef\right]}\left|\theta_{\er}^{\prime}(x)-1\right|\leq \frac{1}{\ef-\er}\leq\frac{1}{\delta}\,.
\label{du theta minus x minus tau kminkmax}
\end{equation}
Using these bounds as well as the inequality $\left|\theta_{\er}(x)-x\right|\leq \er$ and~\eqref{eq: alpha tangent} for $n=1$, it follows
\[
\begin{aligned}
&\left\| \du_{p_{i,\alpha}}\left( \frac{D_{\er,i,1,\alpha'}}{\er} - \mathscr{D}_{i,1,\alpha'} \right) \right\|_{L^{\infty}\left(\mathscr{C}_{i,\ef}\right)}\leq \left(C^*+\frac{2\ef}{m_i\delta }\left\|f^{\prime\prime}_{0,\ef}\right\|_{L^{\infty}} +  C_{\delta} \left\|f_{0,\ef}^{\prime}\right\|\right)\er\,.
\end{aligned}
\]
This concludes the proof of~\eqref{eq: DkimKim esti Kmin2} for $r=1$.

Bounds on higher order derivatives are obtained in a similar fashion, relying on the fact that $\du^2_x\theta_{\er}(x) = 0$ except at the singularity points $\er,\ef$ as well as $\du^2_x\tilde{\tau}(x)=0$ for $x \neq \ef$.
\end{proof}

We end this section with a last technical result.

\begin{lemma}
\label{lemma AdmuKim new}
Fix $\ef > 0$. Then for any $f\in L^{1}\left(\mu_0\right)$, there exist $a_f \in \mathbb{R}$ such that, for $0< \delta <\ef$,
\begin{equation}
\forall \er\in \left[0,\ef-\delta\right],
\qquad
\int_{\mathbb{R}^d} f\left(p\right)\mathrm{e}^{-\beta U_{\er}(p)}dp=\int_{\mathbb{R}^d} f\left(p\right)\mathrm{e}^{-\beta U_{0}(p)}dp+ a_f \er+\mr{O}\left(\er^2\right)\,.
\label{eq: AdmuKim new}
\end{equation}
\end{lemma}

\begin{proof}
Recall that $U_{\er}(p)=\sum_{i=1}^N u_{\er,\ef}(p_i)$. Note that
\[
u_{\er,\ef}(p_i)-u_{0,\ef}(p_i)=\frac{p_i^2}{2m_i}\left[f_{0,\ef}\left(\frac{p_i^2}{2m_i}\right)-f_{\er,\ef}\left(\frac{p_i^2}{2m_i}\right)\right]\,.
\]
Manipulations similar to the ones used to prove~\eqref{eq: DkimKim esti Kmin2} allow to show that there exists a function $\mathcal{U}\in C^{\infty}$ with compact support in $\mathscr{C}_{\ef}$ such that, for $0< \delta <\ef$ and $r\in\mathbb{N}$, there is $C_{r,\delta}>0$ for which
\begin{equation}
\er\in \left[0,\ef-\delta\right], \quad \left\|U_{\er}-U_0-\er\mathcal{U}\right\|_{W^{r,\infty}} \leq C_{r,\delta} \er^2\,.
\label{eq: Ukmin minus U W r infty}
\end{equation}
This allows to write
\[
U_{\er}=U_0+\er\mathcal{U}+\er^2\widetilde{\mathcal{U}}_{\er},
\]
with $\widetilde{\mathcal{U}}_{\er}$ uniformly bounded in $L^{\infty}$. Moreover since $U_{\er}-U_{0}\in C^{\infty}$ also has a compact support, we easily obtain
\[
\frac{\mr{e}^{-\beta U_{\er}} -\mr{e}^{-\beta U_{0}}}{\er}=\mr{e}^{-\beta U_{0}}\frac{\mr{e}^{-\beta \er\left(\mathcal{U}+\er\tilde{\mathcal{U}}_{\er}\right)}-1}{\er}=-\beta\mathcal{U} \, \mr{e}^{-\beta U_{0}} + \er \, \widehat{\mathcal{U}}_{\er} \, \mr{e}^{-\beta U_{0}},
\]
with $\widehat{\mathcal{U}}_{\er}$ uniformly bounded in $L^{\infty}$. Therefore, there exists a constant $R >0$ such that
\[
\left|\frac{\dps \int_{\R^d}f\mr{e}^{-\beta U_{\er}}dp -\int_{\R^d}f\mr{e}^{-\beta U_{0}}dp}{\er}+ \beta \int_{\R^d}f\, \mathcal{U}\,\mr{e}^{-\beta U_{0}}dp\right|\leq R\er\,,
\]
so that~\eqref{eq: AdmuKim new} follows with $\ds a_f:= -\beta \int_{\mathbb{R}^d} \mathcal{U} f\, \mr{e}^{-\beta U_0} \, dp$.
\end{proof}

%-----------------------------------------------------------------
\subsubsection{Verification of Assumption \ref{hyp THM}}
\label{proof of lemma arps condition}

In order to use Lemma~\ref{main theorem kopec}, we need to check that Assumption~\ref{hyp THM} holds with $G_{\nu}$ as small as wanted for appropriate values of $\er,\ef$. The first condition~\eqref{eq: assummption Linfty du2} is easy to check, so we concentrate on the last two conditions.  The reference kinetic energy function $U_\nu$ in Lemma~\ref{main theorem kopec} is chosen as the standard kinetic energy $U_{\mr{std}}(p) = p^T M^{-1}p/2$, so that $\nu=1/ \min_{i=1, \ldots, N}$. It therefore remains to check the last condition. An inspection of the proof of Lemma~\ref{main theorem kopec} reveals that it holds provided $\er,\ef$ are such that~\eqref{eq: def rho m} holds. Straightforward computations show that
\[
\n_{p_i} \left(U_{\er}-U_{\mr{std}}\right) = \frac{p_i}{m_i} \left[1 - f_{\er,\ef}\left(\frac{p_i^2}{m_i}\right)\right] - \frac{p_i |p_i|^2}{m^2_i} f'_{\er,\ef}\left(\frac{p_i^2}{m_i}\right),
\]
so that, using the fact that $U_{\er,\ef}-U_{\mr{std}}$ has compact support in $\mathscr{C}_{\ef}$ (hence $|p_i| \leq\sqrt{2m_i \ef}$) and in view of the expression~\eqref{eq:derivative_f_Kmin} of $f'_{\er,\ef}$, the following bound holds:
\[
\left\|\n_{p_i} \left(U_{\er,\ef}-U_{\mr{std}}\right)\right\|_{L^{\infty}}\leq \sqrt{\frac{2\ef}{m_i}}+\sqrt{ \frac{8\ef^3}{m_i} } \, \frac{\ef}{\ef-\er} \left\|f^{\prime}_{0,\ef}\right\|_{L^{\infty}}.
%\label{eq: 45}
\]
Similarly, there exists a constant $C > 0$ (depending on $f^{\prime}_{0,\ef}, f^{\prime\prime}_{0,\ef}$ and $m_1,\dots,m_N$) such that
\[
\left\|\de U_{\er,\ef}\right\|_{L^{\infty}}\leq C\left[1+\left(\frac{\ef^2}{\ef-\er}\right)^2\right].
%\label{eq: 46}
\]
It is then easy to see that~\eqref{eq: def rho m} holds upon choosing $0 < \er \leq \ef/2$ with $\ef > 0$ sufficiently small.

\subsubsection{Proof of Lemma \ref{corollary: estimate LopD n}}
\label{proof of Lemma estimate LopD n}

Denote by $\mathcal{A}$ the operator~$\LopZero^{-1}\Pi_0\left( D_{\er}\cdot \Lpert\right)$. By Corollary~\ref{lemma Ueps minus Uzero intro}, for any $n\geq 0$ and $0\leq\er\leq\ef/2$, there exists a constant $R_n>0$ such that
\[
\forall \left|\alpha\right|\leq n,
\qquad \left\|\du^{\alpha}_{p}D_{\er}\right\|_{L^{\infty}} \leq R_n \, \er .
\label{eq D est}
\]
By the resolvent estimate~\eqref{eq thm uniform convergence}, there exists for any $s\in\mathbb{N}^*$ a constant $C_s>0$ such that
\[
\forall f\in \widetilde{L^{\infty}_{\K_s}},
\qquad
\left\|\LopZero^{-1}f\right\|_{L^{\infty}_{\K_s}}\leq C_s \left\|f\right\|_{L^{\infty}_{\K_s}}.
\]
Therefore, choosing an integer $s$ for which $A \in W^{1,\infty}_{\K_s}$, there exists a constant $C>0$ such that
\[
\ds
\left\|\mathcal{A} \left(\Pi_{0} A\right)\right\|_{L^{\infty}_{\K_s}}\leq C_s \left\|D_{\er}\cdot \Lpert \left(\Pi_{0}A\right)\right\|_{L^{\infty}_{\K_s}} \leq C \left\|D_{\er}\right\|_{L^{\infty}}\left\| A\right\|_{W^{1,\infty}_{\K_s}}\leq CR_0 \, \er\left\| A\right\|_{W^{1,\infty}_{\K_s}}.
\]
By the same principle, using the fact that, by~\eqref{eq: DkimKim esti Kmin2}, there is for any $r\geq 0$ a constant $C_r>0$ such that	
\[
\left\|D_{\er}\right\|_{W^{r,\infty}}\leq C_r\er,
\]
and in view of~\eqref{eq: functional estimates Kopec_resolvent}, there exists, for any $l\geq 0$, integers $\alpha \geq l$ and $s_l\in\mathbb{N}$ such that, for all $s\geq s_l$, there is a constant $C>0$ and an integer $r\in\mathbb{N}$ for which
\[
\ds
\left\|\mathcal{A} \left(\Pi_{0} A\right)\right\|_{W^{l,\infty}_{\K_s}}\leq C \er \left\| A\right\|_{W^{\alpha,\infty}_{\K_r}}.
\]
By recurrence, there exist, for any $n \geq 1$, integers $s_n,l_n \geq 0$ such that, for all $s\geq s_n$, there is $r \in\mathbb{N}$ and $\widetilde{C}>0$ for which
\[
\left\|\mathcal{A}^{n} \left(\Pi_{0} A\right)\right\|_{L^{\infty}_{\K_s}}\leq \widetilde{C} \er^{n} \left\| A\right\|_{W^{l_n,\infty}_{\K_{r}}}\,.
\]
This gives the claimed result.

\subsubsection{Proof of Lemma \ref{theorem truncated variance}}
\label{proof of theorem truncated variance}

We start by writing the difference between the variance~\eqref{eq: variance F} and the truncated one~\eqref{def truncated variance}:
\begin{equation}
\sigma^2_A(\er)-\sigma^{2}_{A,n}(\er) = -2\int \left(\Phi_{A,\er}-\Pi_{\er}\Phi_{A,\er}^n\right) A \, d\mu_{\er} \,.
\label{eq: variance difference}
\end{equation}
A simple computation gives
\begin{equation}
\Pi_{\er} \Lop_{\er} \left(\Phi_{A,\er}-\Phi_{A,\er}^n\right) =-\Pi_{\er}\left(D_{\er}\cdot\Lpert\right) \left(\Lop_0^{-1}\Pi_0D_{\er}\cdot\Lpert\right)^n \Lop_0^{-1}\Pi_0A \,.
\label{eq: pi Lope F-Fe}
\end{equation}
We first use Lemma~\ref{corollary: estimate LopD n}: there exists $s_n,l_n\in\mathbb{N}$ such that, for $s\geq s_n$,  there is $r_n\in\mathbb{N}$ and $C>0$ such that
\[
\left\| \left(\LopZero^{-1}\Pi_0 D_{\er}\cdot \Lpert\right)^n \LopZero^{-1}\Pi_0 A \right\|_{W^{1,\infty}_{\K_{s}}} \leq C \er^{n} \left\|A\right\|_{W^{l_n,\infty}_{\K_{r_n}}}\, .
\]
Therefore, using~\eqref{eq: Ueps minus Uzero intro} and~\eqref{eq: pi Lope F-Fe}, there exists some constant $R_n>0$ such that
\[
\left\|\Pi_{\er}\Lop_{\er}\left(\Phi_{A,\er}-\Phi_{A,\er}^n\right)\right\|_{L^{\infty}_{\K_s}} \leq R_n\er^{n+1} \left\| A\right\|_{W^{l_n,\infty}_{\K_{r_n}}} .
\label{eq: 33}
\]
We finally apply $\Lop_{\er}^{-1}$ to both sides of~\eqref{eq: pi Lope F-Fe}: in view of~\eqref{eq: corollary convergence thm}, it follows
\[
  %\label{eq: 2 variance difference}
  \left\|\Pi_{\er}\left(\Phi_{A,\er}-\Phi_{A,\er}^n\right)\right\|_{L^{\infty}_{\K_s}}\leq \frac{R_nC_s}{\lambda_s}\er^{n+1} \left\| A\right\|_{W^{l_n,\infty}_{\K_{r_n}}} .
\]
The result is then a direct consequence of the equality~\eqref{eq: variance difference}.

\subsubsection{Proof of Proposition \ref{variance perturbation main theorem}}
\label{proof of variance perturbation main theorem}
Looking at \eqref{def truncated variance}, there are three objects which depend on the parameter $\er$: the projection $\Pi_{\er}$, the truncated solution of the Poisson equation $\Phi_{A,\er}^n$ and the modified measure $\mu_{\er}$.

We first expand $\Phi^{n}_{A,\er}$ in terms of $\Phi_{A,0}$ as
\[
\begin{aligned}
  \Phi_{A,\er}^n
  &=\Lop_0^{-1}\Pi_0A + \left(\Lop_0^{-1}\Pi_0D_{\er}\cdot\Lpert\right) \Lop_0^{-1}\Pi_0A +\sum_{k=2}^{n}\left(\Lop_0^{-1}\Pi_0D_{\er}\cdot\Lpert\right)^k \Lop_0^{-1}\Pi_0A \\
  &=\Phi_{A,0} + \left(\Lop_0^{-1}\Pi_0D_{\er}\cdot\Lpert\right) \Phi_{A,0}+\sum_{k=2}^{n} \left(\Lop_0^{-1}\Pi_0D_{\er}\cdot\Lpert\right)^k  \Phi_{A,0} \,.
\end{aligned}
\]
Estimates on $\Phi_{A,0}$ and its derivatives in terms of~$A$ can be obtained with~\eqref{eq: functional estimates Kopec_resolvent}. Lemma~\ref{corollary: estimate LopD n} then allows to estimate the higher order terms in the above equality: there exists $s \in \mathbb{N}$ and $C > 0$ such that
\[
\left\| \Phi_{A,\er}^n - \Phi_{A,0} - \left(\Lop_0^{-1}\Pi_0D_{\er}\cdot\Lpert\right) \Phi_{A,0} \right\|_{L^\infty_{\K_s}} \leq C \, \er^2.
\]
By combining these estimates with \eqref{eq: DkimKim esti Kmin2}, we obtain
\begin{equation}
 \Phi_{A,\er}^n = \Phi_{A,0} +\er \left(\Lop_0^{-1}\Pi_0\mathscr{D}\cdot\Lpert\right) \Phi_{A,0}+ \er^2 \mathcal{R}_{\er},
\label{eq: expansion phiAern}
\end{equation}
where $\mathcal{R}_{\er}$ is uniformly bounded in $L^{\infty}_{\K_s}$ due to \eqref{eq: DkimKim esti Kmin2} for $\er$ small enough (upon possibly increasing~$s$).

With the notation of Lemma~\ref{lemma AdmuKim new}, for any $f\in L^{1}\left(\mu_0\right)$,
\begin{equation}
\label{eq: <f>dmu expansion}
\ds \int_{\El}f \, d \mu_{\er} = \frac{\ds \int_{\El}f \mr{e}^{-\beta U_{\er}}}{\ds \int_{\El} \mr{e}^{-\beta U_{\er}}} = \int_{\El}f \, d\mu_0 + \frac{\dps a_f -a_1\int_{\El}f \, d\mu_0}{\dps \int_{\El} \mr{e}^{-\beta U_{0}} }\, \er + \widetilde{\mathscr{R}}_{\er}\er^2,
\end{equation}
with $\widetilde{\mathscr{R}}_{\er}$ uniformly bounded for $\er$ small enough.
Finally, by combining~\eqref{eq: expansion phiAern} and~\eqref{eq: <f>dmu expansion}, we see that there exists $\mathscr{K}\in \R$ such that
\begin{equation}
\begin{aligned}
\sigma^2_{A,n}(\er) &= -2\int_{\El} \Phi_{A,0} A  \, d\mu_{0}+\mathscr{K}\er+\mathrm{O}\left(\er^2\right).
\end{aligned}
\end{equation}

%---------------- REMERCIEMENTS ------------------
%\subsection*{Acknowledgements}
\begin{acknowledgements}
 Stephane Redon and Zofia Trstanova gratefully acknowledge funding from the European Research Council through the ERC Starting Grant n. 307629. This work was funded by the Agence Nationale de la Recherche, under grant ANR-14-CE23-0012 (COSMOS). Gabriel Stoltz benefited from the scientific environment of the Laboratoire International Associ\'e between the Centre National de la Recherche Scientifique and the University of Illinois at Urbana-Champaign.
\end{acknowledgements}

%\begin{acknowledgements}
%\end{acknowledgements}

%\bibliographystyle{aps-nameyear}      % American Physical Society (APS) style, author-year citations
%\bibliography{MyBibliography}                % name your BibTeX data base

%   \input{article_JStatPhys.bbl}

%\bibliography{article_JStatPhys}

%%%\begin{thebibliography}{}
%%%%
%%%% and use \bibitem to create references. Consult the Instructions
%%%% for authors for reference list style.
%%%%
%%%% Format for Journal Reference
%%%\bibitem[\protect\citeauthoryear{Aamport}{1986}]{RefJ}
%%%L.A. Aamport, \mbox{G-Animal's} Journal \textbf{41} (7), 73 (1986).
%%%This is a full ARTICLE entry
%%%
%%%% Format for books
%%%\bibitem[\protect\citeauthoryear{Knuth}{1981}]{book-full}
%%%D.E. Knuth, \textit{Seminumerical algorithms}, 2nd edn.
%%%The Art of Computer Programming,
%%%vol. 2 (Addison-Wesley, Reading, 1981).
%%%This is a full BOOK entry
%%%
%%%
%%%% Format for proceedings
%%%\bibitem[\protect\citeauthoryear{Oz and Yannakakis}{1983}]{RefB}
%%%W.V. Oz, M. Yannakakis (eds.),
%%%in \textit{All ACM Conferences} (Academic Press, Boston, 1983).
%%%This is a full PROCEEDINGS entry
%%%% Other formats available: INPROCEEDINGS, PHDTHESIS, TECHREPORT,
%%%% UNPUBLISHED, MISC, MASTERSTHESIS, MANUAL, INCOLLECTION, BOOKLET
%%%% etc
%%%\end{thebibliography}

\end{document}